\documentclass[reqno,12pt,letterpaper]{amsart}
\usepackage{amsmath,amssymb,amsthm,graphicx,mathrsfs,url}\usepackage[usenames,dvipsnames]{color}
\usepackage[colorlinks=true,linkcolor=Red,citecolor=Green]{hyperref}
\usepackage{amsxtra,bbm}
\usepackage{subcaption}
\usepackage{epstopdf}
\def\arXiv#1{\href{http://arxiv.org/abs/#1}{arXiv:#1}}

\newcommand{\ket}[1]{|#1\rangle} 
\newcommand{\bra}[1]{\langle#1|} 
\newcommand{\kb}[1]{|#1\rangle\langle#1|} 

\def\?[#1]{\textbf{[#1]}\marginpar{\Large{\textbf{??}}}}
\def\smallsection#1{\smallskip\noindent\textbf{#1}.}
\let\epsilon=\varepsilon 

\newcommand{\RR}{{\mathbb R}}
\newcommand{\NN}{{\mathbb N}}
\newcommand{\R}{\mathbbm{R}}
\newcommand{\cE}{\mathcal{E}}

\newcommand{\C}{\mathbbm{C}}
\newcommand{\N}{\mathbbm{N}}
\newcommand{\CC}{{\mathbb C}}
\newcommand{\cD}{\mathcal{D}}

\newtheorem{example}{Example}
\newcommand{\F}{\mathcal{F}}

\newcommand{\nn}{\nonumber}
\newcommand{\comment}[1]{}
\newcommand{\ran}{\operatorname{ran}}

\newtheorem{theorem}{Theorem}
\newtheorem{prop}{Proposition}[section]

\newtheorem{ass}{Assumption}
\newcommand{\eps}{\varepsilon}
\newtheorem{lemm}[prop]{Lemma}
\newtheorem{corr}[prop]{Corollary}
\newtheorem{rem}{Remark}
\numberwithin{equation}{section}

\newcommand{\1}{\mathbbm{1}}
\DeclareMathOperator{\Spec}{Spec}

\newcommand{\cB}{{\mathcal{B}}}

\newcommand{\cL}{\mathcal{L}}

\newcommand{\cH}{\mathcal{H}}
\newcommand{\cT}{\mathcal{T}}

\let\Re=\Real

\DeclareMathOperator{\tr}{tr}

\usepackage{scalerel}

\newcommand{\ro}[1]{{\color{BurntOrange} #1}}

\newcommand\reallywidehat[1]{\arraycolsep=0pt\relax%
\begin{array}{c}
\stretchto{
  \scaleto{
    \scalerel*[\widthof{\ensuremath{#1}}]{\kern-.5pt\bigwedge\kern-.5pt}
    {\rule[-\textheight/2]{1ex}{\textheight}} 
  }{\textheight} %
}{0.5ex}\\           
#1\\                 
\rule{-1ex}{0ex}
\end{array}
}

\title[Quantum Zeno effect for open quantum systems]{Quantum Zeno effect for open quantum systems}
\author{Simon Becker}
\email{simon.becker@damtp.cam.ac.uk}
\address{University of Cambridge,
DAMTP, Wilberforce Rd, Cambridge CB3 0WA, UK}
\author{Nilanjana Datta}
\email{n.datta@damtp.cam.ac.uk}
\address{University of Cambridge,
DAMTP, Wilberforce Rd, Cambridge CB3 0WA, UK}
\author{Robert Salzmann}
\email{rs2059@maths.cam.ac.uk}
\address{University of Cambridge,
DAMTP, Wilberforce Rd, Cambridge CB3 0WA, UK}
\begin{document}
\begin{abstract}
We prove the quantum Zeno effect in open quantum systems whose evolution, governed by quantum dynamical semigroups, is repeatedly and frequently interrupted by the action of a quantum operation. For the case of a quantum dynamical semigroup with a bounded generator, our analysis leads
to a refinement of existing results and extends them to a larger class of quantum operations. We also prove the existence of a novel strong quantum Zeno limit for quantum operations for which a certain spectral gap assumption, which all previous results relied on, is lifted. The quantum operations are instead required to satisfy a weaker property of strong power-convergence. In addition, we establish, for the first time, the existence of a quantum Zeno limit for the case of unbounded generators. We also provide a variety of physically interesting examples of quantum operations to which our results apply.
\comment{We then derive novel quantum Zeno dynamics for open systems whose time evolution is governed by unbounded generators. }
\end{abstract}
\maketitle
	\section{Introduction}
	The quantum Zeno effect describes the phenomenon that frequently measuring a quantum system slows down its time evolution and eventually freezes it completely. The effect has been named after the Greek philosopher Zeno who introduced an argument for the paradox that a flying arrow which is continuously observed cannot move and therefore never reaches its target.
	The quantum Zeno effect was theoretically predicted by Misra and Sudarshan in 1977 \cite{MisSud77} and experimental verification of the phenomenon was achieved in \cite{ItHeBoWi90,FiGuRai01}. 

	Besides its striking implications for fundamental physics, the quantum Zeno effect has many practical applications, for example in control of decoherence \cite{FJP04,HRBPK06}, quantum error correction \cite{EARV04,PSRDL12} and state preparation \cite{NTY03,NUY04,WYN08}.

    \medskip	
	
	Let us first consider the quantum Zeno effect for closed quantum systems. We associate to a closed quantum system a separable, possibly infinite-dimensional, complex Hilbert space, $\cH$, and the time evolution of the system is governed by Schrödinger's equation, which under suitable choice of units ($\hbar = 1$) can be written as 
	\begin{align}
\left\{
	\begin{aligned}
	i\partial_t\psi(t) &= H\psi(t)\nonumber \\
	\psi(0) &= \psi_0.
	\end{aligned}
	\right.
\end{align}
	Here $H$ denotes the Hamiltonian governing the dynamics of the system, and is a self-adjoint operator on $\cH$. The solution of Schrödinger's equation is given by $\psi(t)= e^{-itH}\psi_0$, with $\left(e^{-itH}\right)_{t\in\R}$ being the unitary group generated by $H$.
	
	In the simplest setup, the quantum Zeno effect for closed quantum systems can be formalized in the following way: The system starts in a pure state corresponding to some $\psi_0 \in\cH$ at time zero. For $t>0$ being the total time of the experiment and $n\in\N$, the system evolves for a time $t/n$ under Schrödinger's evolution and is then subjected to a binary von Neumann (i.e.~projective) measurement corresponding to the projections $\{\kb{\psi_0},\1-\kb{\psi_0}\}$. This process is repeated $n$ times. The quantum Zeno effect predicts that the probability, $p_n$, of always finding the system in the initial state $\psi_0$ (and thus with measurement outcome $\kb{\psi_0}$), converges to $1$ in the limit $n\to\infty$, i.e., 
	\begin{align*}
	p_n = \Big\|\big(\kb{\psi_0}e^{-itH/n}\big)^n\psi_0\Big\|^2 = \big|\langle\psi_0,e^{-itH/n}\psi_0\rangle\big|^{2n} \xrightarrow[n\to\infty]{}1.
	\end{align*}
	Hence, even though $\psi_0$ might not be an invariant state under Schrödinger's evolution, measuring the system frequently enough will freeze the system in the state $\psi_0$ in the limit of asymptotically many measurements.
	
	More generally, one can consider projective measurements $\{P,\1- P\}$ (where $P$ is a general projection operator) and mixed initial states. Given an initial state $\rho_0$, the probability of always obtaining the measurement outcome corresponding to $P$ when measuring the system repeatedly in time intervals of size $t/n$ is given by
	\begin{align}
	p_n = \tr\Big((Pe^{-itH/n})^n\rho_0(e^{itH/n}P)^n\Big).
	\end{align}
In this setup the quantum Zeno effect manifests itself in the convergence of this probability (which we call the {\em{survival probability}}) to the expectation value of $P$ in the initial state $\rho_0$ , i.e. $\lim_{n\to\infty}p_n = \tr(P\rho_0)$.

	Apart from the convergence of the survival probability, $p_n$, one might also be interested in the effective dynamics emerging from the process of repeatedly measuring the evolving system. Under the assumption that the quantum Zeno effect occurs and the initial state satisfies $\tr(P\rho_0)=1$, the only non-trivial part of the effective time evolution takes place in the invariant subspace of the projection $P$, which is referred to as the {\em{quantum Zeno subspace}}. In the following, we refer to $P$ as the {\em{quantum Zeno projection}}. Formally it is expected that the effective dynamics within this subspace (in the limit of asymptotically many periodic measurements) which is called the {\em{quantum Zeno dynamics}}, is given by
	\begin{align}
	\label{eq:UnitaryZenoDynamics}
	(Pe^{-itH/n})^n \xrightarrow[n\to\infty]{}e^{-itPHP}P.
	\end{align}
	Here, for $H$ being an unbounded operator, the expression $PHP$ can only be understood in the formal sense  and one needs to find a rigorous definition of the right self-adjoint operator which is the generator of the effective unitary time evolution. It is important to note that by unitarity of the effective time evolution in the quantum Zeno subspace one can infer the quantum Zeno effect from the quantum Zeno dynamics, i.e.
	\begin{align*}
	\lim_{n\to\infty}p_n = \tr\Big(e^{-itPHP}P\rho_0Pe^{itPHP}\Big) = \tr(P\rho_0).
	\end{align*}
	\medskip
	
	While the quantum Zeno effect has been thoroughly studied for closed quantum systems (i.e.~on Hilbert spaces), for both bounded and unbounded Hamiltonians (see \cite{ExIch05} or \cite{FP} and references therein for a review), the results on semigroups on Banach spaces, which is the right setup for open quantum systems, have been mostly restricted to semigroups generated by bounded operators \cite{MaShvi03,Ma04,G,MobWolf19}.  In this paper we extend the study of the quantum Zeno effect to open quantum systems whose dynamics is generated by unbounded operators. We also make a more refined analysis for the case of bounded generators, thus improving on existing results.
	
	\noindent
\subsection{Quantum Zeno effect for open quantum systems:} As mentioned above, in this work we extend and analyse the quantum Zeno effect and its associated dynamics for open quantum systems.
	As in this case the physical system has unavoidable interactions with its environment, for example a thermal bath in which it is placed, the time evolution of the system is no longer governed by a unitary group on its Hilbert space, $\cH$. If the coupling between the system and its environment is weak, the time evolution can be approximately described by a dynamical semigroup of completely positive, trace-preserving maps on the Banach space, $\cT(\cH)$, of trace-class operators which we denote by $\left(e^{t\cL}\right)_{t\ge0}$, with $\cL$ being the generator; a general discussion of dynamical semigroups can be found in Section~\ref{sec:semigroups}. As in the case of closed systems, we refer to the effective dynamics arising from the process of frequently performing projective measurements on the open system, while letting it evolve under the dynamics given by the semigroup $\left(e^{t\cL}\right)_{t \geq 0}$, as the {\em{quantum Zeno dynamics}}. 
	
	\medskip
	
	For a general Banach space, $X$, Matolcsi and Shvidkoy proved in 2003~\cite{MaShvi03} that for $\cL$ being a bounded linear operator and $P$ being a bounded projection on $X$,
	\begin{align}
	\label{eq:ProductFormulaProjection}
	(Pe^{t\cL/n})^n \xrightarrow[n\to\infty]{} e^{tP\cL P}P.
	\end{align}
	Their motivation for analyzing this limit was to investigate general features of dynamical semigroups, especially as \eqref{eq:ProductFormulaProjection} can be viewed as a degenerated version of the Lie-Trotter product formula $\lim_{n\to\infty}\left(e^{tA/n}e^{tB/n}\right)^n = e^{t(A+B)}$.
	
	\medskip
	
	For an open quantum system governed by a quantum dynamical semigroup $\left(e^{t\cL}\right)_{t\ge0}$, with bounded generator $\cL$, the limit in \eqref{eq:ProductFormulaProjection} yields the desired quantum Zeno dynamics. However, unlike the case of unitary dynamics of closed quantum systems, the effective dynamics given by $e^{tP\cL P}$ is in general not trace-preserving. Hence, we cannot infer the quantum Zeno effect (i.e.~convergence of the survival probabilities), as for closed systems, from the quantum Zeno dynamics itself. This can be seen from the so-called GKLS form~\cite{Lin76,GKS76} for bounded generators $\cL$ of completely positive trace-preserving semigroups, according to which, for any $\rho\in\cT(\cH)$,
	\begin{align}
	\label{eq:GKLS}
	\cL(\rho) = K\rho + \rho K^* + \sum_{l} L_l\rho L_l^*, 
	\end{align}
	\begin{align}
	\label{eq:GKLSsum}
\text{ under the constraint }	K^* + K +\sum_{l}L^*_lL_l  = 0.
	\end{align} 
	The identity \eqref{eq:GKLSsum} ensures that for any $t\ge 0$, $e^{t\cL}$ is trace-preserving.
	Here, the index $l$ might range over an infinite set and the sums in the above equations converge in suitable operator topologies (see~\cite{Chang} for details). 
	
	Consider now a specific form of the quantum Zeno projection operator $P$ on $\cT(\cH)$, which is given by $P(\rho) = \pi \rho \pi$ for some projector $\pi$ on $\cH$. Using \eqref{eq:GKLS} we can find a similar expression for the effective generator, $P\cL P$, of the quantum Zeno dynamics, which acts on any state $\rho$ in the quantum Zeno subspace $P\cT(\cH)$ as follows:
	\begin{align}\label{eq:p-gen}
	P\cL P(\rho) &= \pi K\pi \rho+ \rho (\pi K\pi)^* + \sum_{l} (\pi L_l\pi)\rho (\pi L_l\pi)^*\nn \\ &=  P(K) \rho+ \rho P(K)^* + \sum_{l} P(L_l)\rho P(L_l)^*.
	\end{align}
	By comparing \eqref{eq:GKLS} with \eqref{eq:p-gen} we see that the operators $K$ and $L_l$ in the former are replaced by $P(K)$ and $P(L_l)$ in the latter, and hence $P(K)$ and $P(L_l)$ can be viewed as the corresponding operators in the GKLS form of the generator, $P\cL P$, of the effective dynamics. However, making these replacements on the left hand side of \eqref{eq:GKLSsum} yields an expression which is negative semidefinite. 
	In fact, one can convince oneself that this resulting expression might not be equal to zero by considering the following example: $\cH = \C^2$, $L_0 = \ket{0}\bra{0}$, $L_1 = \ket{0}\bra{1}$ (with $\{\ket{0},\ket{1}\}$ being an orthonormal basis of $\C^2$), $K = -\1/2$ and $\pi = \kb{1}$. Thus we infer that the effective dynamics generated by $P\cL P$ is not trace-preserving. Instead, $P\cL P$ is the generator of a completely positive, {\em{trace non-increasing}} semigroup on $P\cT(\cH)$.
	
	This example shows that for open quantum systems the survival probabilities will not be frozen as for closed systems, i.e. we might have
	\begin{align*}
	\lim_{n\to\infty}p_n = \lim_{n\to\infty}\tr\Big(\left(Pe^{t\cL/n}\right)^n(\rho)\Big) < \tr(P(\rho)).
	\end{align*}
However, in the limit of the number measurements ($n$) tending to infinity, the only non-trivial contribution to the survival probability arises if all successive measurement outcomes are identical -- either all of them corresponding to $P$ or all of them corresponding to $\1-P$. In order to see this, consider for example the probability of the first measurement yielding an outcome corresponding to $P$ and all subsequent ones corresponding to $\1-P$. Let us denote this probability by $p'_n$. We see that
	\begin{align*}
	p_n' = \tr\Big(\left((\1-P)e^{t\cL/n}\right)^{n-1}Pe^{t\cL/n}(\rho)\Big) = \mathcal{O}(1/n),
	\end{align*} 
	where we have used that for $\cL$ bounded $e^{t\cL/n} =\1 + \mathcal{O}(1/n)$ and $(\1-P)P=0$.
	
	\medskip
	
	For open quantum systems, one can not only perform projective measurements but also {\em{generalized measurements}}. These can be described by a collection  $\{M_j\}_{j}$ of quantum operations, i.e.~completely positive, trace non-increasing maps on $\cT(\cH)$, with the subscripts $j$ labelling the outcomes, such that their sum is trace-preserving. Here, the probability of measuring an outcome $j$ given a state $\rho$ is given by $q_j = \tr\left(M_j(\rho)\right)$ and the corresponding post-measurement state is given by $M_j(\rho)/q_j$ for non-zero $q_j$.
	
	For open quantum systems, a more general framework for studying the quantum Zeno effect is one in which the projective measurements are replaced by repeated actions of a fixed quantum operation $M$. The latter acts between individual time intervals of length $t/n$ over which the system evolves under the action of a generator $\cL$ of a dynamical semigroup. This is given by the {\em{quantum Zeno product}}
	\begin{equation}
	\label{eq:ZenoProduct}
	\left(Me^{t\cL/n}\right)^n.
	\end{equation}
	In the sequel, we refer to the asymptotic behaviour of the quantum Zeno product as $n \to \infty$ as the \emph{quantum Zeno limit.\footnote{Henceforth, we often suppress the word 'quantum' for simplicity.}} Recently Möbus and Wolf~\cite{MobWolf19} studied the quantum Zeno effect in this framework, thus extending the general semigroup results of~\cite{MaShvi03,Ma04,G}. 
	They proved convergence of the quantum Zeno product $\left(Me^{t\cL/n} \right)^n$ to an effective quantum Zeno dynamics in infinite dimensions in the case in which $M$ satisfies a certain spectral gap assumption and $\cL$ is bounded. We discuss their result in more detail in Section~\ref{sec:main-results}.
	Independently, Burgath et al.~proved convergence of the quantum Zeno product for general quantum operations $M$ in finite dimensions \cite{BuFaNaPaYu18}.
	\medskip
	
	In this article, for sake of generality, instead of simply focussing on the space $\cT(\cH)$ of trace-class operators (as in~\cite{MobWolf19}), we consider arbitrary Banach spaces $X$. Denoting the set of bounded linear operators on $X$ by $\cB(X)$, we assume that $M \in \cB(X)$ is a contraction, i.e.~its operator norm satisfies the bound $||M|| \leq 1$, and that $\left(e^{t\cL}\right)_{t\ge0}$ generates a contraction semigroup on $X$ which is only assumed to be strongly continuous (see Section~\ref{sec:semigroups} for details on dynamical semigroups).
	
\smallsection{New contributions of this article:}	
	 In this article, we extend the analysis of the quantum Zeno effect for open quantum systems in \cite{MobWolf19} in multiple ways: we provide a quantitative version of the quantum Zeno limit derived in~\cite{MobWolf19} and identify a more general condition on the spectrum of the quantum operation $M$ which is both necessary and sufficient for the Zeno product to be norm convergent to an effective Zeno dynamics. 
	 This is given in Proposition~\ref{lem:SpectralGap}. In particular, our condition shows that apart from a spectral gap condition, there must be no (quasi)-nilpotent contribution to the eigenspaces of the quantum operation $M$ on the unit circle. Such an assumption was missing in~\cite{MobWolf19}. Currently it is not known 
	 whether such an assumption always holds for general quantum channels acting on infinite-dimensional quantum systems, as is the case for finite-dimensional quantum systems~\cite{Wolf12}, or whether it has to be additionally imposed. In our framework for general Banach spaces and contraction maps $M$, we show that this condition cannot be omitted.
	 
	 \medskip
	 
	 In addition, we derive, for the first time, a quantum Zeno limit for unbounded generators by combining the quantitative result for bounded generators with bounded \emph{Yosida approximations} of the unbounded generator.
	 
	 \medskip
	 
	  We also go beyond the ubiquitous spectral gap assumption for $M$ and show that there still exists a strong quantum Zeno limit if the spectral gap assumption in  \cite{MobWolf19} is omitted and replaced by a strong power-convergence property of $M$. This relies on an entirely new perturbation series approach towards the quantum Zeno effect. We complete our new approach by identifying a variety of sufficient conditions and physical examples of quantum channels $M$ that satisfy the strong power-convergence property.
	 
Finally, we illustrate our findings by studying various concrete examples of quantum channels.

\smallsection{Outline of the article:}	
	 This article is organized in the following way: 
	 \begin{itemize}
	\item In Section~\ref{sec:Mathprelim} we review facts about dynamical semigroups on Banach spaces and spectral projections. 
	 \item In Section \ref{sec:main-results} we present our main results, given by  Theorem~\ref{thm:QuantitativQuantumZeno}, Theorem~\ref{thm:StrongZenoGeneral} and  Theorem~\ref{thm:UnboundedQuantumZenoProjector}. 
	 \item In Section \ref{sec;Prop-examples}, we review some basic facts about operator ergodic theory which we employ in our proofs, and state the proof of Proposition~\ref{lem:SpectralGap}.
	 \item In Section~\ref{sec:QuantitativeBoundedZeno} we prove Theorem~\ref{thm:QuantitativQuantumZeno}, namely, the convergence of the quantum Zeno product $\left(Me^{t\cL/n}\right)^n$ for bounded generators. 
	 \item In Section \ref{sec:UnboundedQuantumZenoProjector} we prove Theorem~\ref{thm:UnboundedQuantumZenoProjector}, namely, the convergence of the quantum Zeno product $\left(Me^{t\cL/n}\right)^n$ for unbounded generators. 
	 \item In Section \ref{sec:withoutspecgap} we prove Theorem \ref{thm:StrongZenoGeneral} which states that the spectral gap condition on $M$ can be replaced by a strong power-convergence property, for a strong quantum Zeno limit to hold.  
	 \item In Section~\ref{sec:strongpowconv}, we discuss two ergodic methods to prove strong power-convergence to an invariant state for quantum dynamical semigroups. This provides a plethora of further examples for Theorem \ref{thm:StrongZenoGeneral}.
	 \item Finally, we\comment{provide a summary of our results and} state some open problems in Section~\ref{sec:open}.
	 \end{itemize}
	\section{Mathematical Preliminaries}
	\label{sec:Mathprelim}
	\smallsection{Notation}
	Let $X$ denote a Banach space, and $\cB(X)$ be the set of bounded linear operators on it.  In particular, let $\1 \in \cB(X)$ denote the identity operator acting on $X.$  For a bounded linear operator $T\in\cB(X)$ we write $\ker(T)$ to denote the nullspace and $\operatorname{ran}(T)$ to denote the range or image. We call $T$ a contraction if $\|T\|\le1$, where $\|\cdot\|$ denotes the operator norm on $\cB(X).$
	
	\medskip
	
	A complex number $\lambda \in {\mathbb{C}}$ is said to be in the {\em{resolvent set}}, $\rho(T)$, if $(\lambda \1 - T) \in \cB(X)$ is a bijection. For $\lambda \in \rho(T)$, the operator $R_\lambda(T) := (\lambda - T)^{-1}\in\cB(X)$ is called the {\em{resolvent}} and is well-defined. Here and henceforth, $(\lambda - T)$ denotes $(\lambda \1 - T)$. The {\em{spectrum}} of $T$, denoted as $\Spec(T)$ is the complement of the resolvent set. The {\em{spectral radius}} of $T$ is the radius of the smallest disc centered at the origin which contains $\Spec(T)$: $r(T) := \sup \{ |\lambda| \,: \, \lambda \in \Spec(T)\}$. In addition, Gelfand's formula holds
	\begin{equation}\label{spec-radius}
	    r(T) = \lim_{n \to \infty} ||T^n||^{1/n}.
	\end{equation}
	The spectrum $\Spec(T)$ of an operator $T\in\cB(X)$ can be decomposed into three disjoint parts:
	
	$(i)$ Point spectrum:
	\begin{align*}
	    \Spec_p(T) := \left\{\lambda \in \C \, :\, Tx = \lambda x \text{ for some } 0\neq x\in X\right\}.
	\end{align*}
	Each $\lambda$ in $\Spec_p(T)$ is said to be an eigenvalue of $T$ and each $0\neq x\in X$ with $Tx=\lambda x$ is called eigenvector corresponding to $\lambda$.
	
	$(ii)$ Continuous spectrum: The continuous spectrum consists of all $\lambda \notin \Spec_p(T) $ such that $\lambda - T$ is not surjective and $\ran(\lambda - T)$ is dense in $X$.
	
	$(iii)$ Residual spectrum: If $\lambda \not\in \Spec_p(T)$ and $\ran(\lambda - T)$ is not dense, then $\lambda$ is said to be in the residual spectrum of $T$.
	
	\medskip
	
In the context of the quantum Zeno effect, the most relevant Banach space is that of trace-class operators $\cT(\cH)$ on some separable Hilbert space $\cH$. Density operators (or quantum states) $\rho \in \cT(\cH)$ are positive trace-class operators of unit trace. An operator $T\in\cB(\cT(\cH))$ is completely positive if 
	\begin{align*}
	\left(T\otimes\1_d\right)\left(\rho\right)\ge 0,\quad\quad \forall d\in\N,\, \rho\in\cT(\cH)\otimes\C^{d\times d},\,\text{with }\rho\ge 0,\quad
	\end{align*}
	where we have denoted the identity map on the $d$-dimensional complex square matrices $\C^{d\times d}$ by $\1_d$. Moreover, $T\in\cB(\cT(\cH))$ is trace-preserving if for all $x\in\cT(\cH)$ we have $\tr(T(x)) = \tr(x)$ and trace non-increasing if $\tr(T(x)) \le \tr(x)$ for all $x\ge 0$. We call a linear, completely positive operator $T\in\cB(\cT(\cH))$ a {\em{quantum operation}} if it is trace non-increasing, and a {\em{quantum channel}} if it is trace-preserving. Note that every quantum operation is a contraction.
	Further, we denote by $\operatorname{HS}(\mathcal H)$ the Hilbert space of Hilbert-Schmidt operators acting on $\mathcal{H}$.
	
\comment{We now recall some general concepts from semigroup theory, which we employ in our proofs.}

	\subsection{Dynamical semigroups}\label{sec:semigroups}
	In the following we recall some general concepts from semigroup theory (see~\cite{EngNag00} for more details).
	Let $X$ be a Banach space: we say $(T(t))_{t\ge0} \subset \cB(X)$ is a one-parameter semigroup if
	\begin{enumerate}
	    \item $T(t)T(s) = T(t+s),$ for all $t,s\ge0,$
	    \item $T(0) = \1$.
	\end{enumerate}
The one-parameter semigroup is said to be {\em{uniformly- or norm continuous}} if $\lim_{t\downarrow 0}\|T(t) - \1\| = 0$. On the other hand, a semigroup is {\em{strongly continuous}} if for all $x\in X$ we find $\lim_{t\downarrow 0}\|(T(t) - \1)x\| = 0. $
	For any such semigroup, we can define the densely-defined and closed generator $\cL$ by
	\begin{align}
	\label{eq:DefGenerator}
	\cL x = \lim_{t\downarrow 0} \frac{T(t) - \1}{t\,}x
	\end{align}
	for all $x$ in the domain $\cD(\cL)\subseteq X$, which is the set of $x$ for which the strong limit on the right hand side of \eqref{eq:DefGenerator} exists. The generator is bounded if and only if the semigroup is uniformly continuous, in which case $T(t) = e^{t\cL}$.
	For contraction semigroups, i.e.~semigroups satisfying $\|T(t)\|\le 1$ for all $t\ge0$, we can recover the semigroup from its generator as follows. The spectrum of $\cL$ is contained in the left half plane of $\CC$ and in addition the resolvents satisfy the bound \cite[Theo $3.5$]{EngNag00}
	\begin{equation}
	\label{eq:ResolventBound}
	 \|\lambda\left(\lambda -\cL\right)^{-1}\| \le 1 \text{ for all }\lambda >0.
	\end{equation}
    Hence, for each $k\in\N$ we can define the $k^{th}$ \emph{Yosida approximant} of the generator by
	\begin{equation}
	\label{eq:YosidaOperator}
	\cL_k = k\cL\left(k-\cL\right)^{-1},
	\end{equation}
	which are bounded operators satisfying $\|\cL_k\| \le k$ and in addition
	\begin{align*}
	   \cL_kx \xrightarrow[k\to\infty]{} \cL x \text{ for all }x\in\cD(\cL).
	\end{align*}
         From the Yosida approximants, the semigroup can be recovered as the strong limit
	\[ \lim_{k\to\infty} e^{t\cL_k}x = T(t)x=: e^{t\cL}x \quad \forall \, x \in \cB(X).\]

	As mentioned earlier in this article, we mainly consider the Banach space $X$ to be the space of trace-class operators $\cT(\cH)$ on some separable Hilbert space $\cH$, and each $T(t)$ for $t\ge0$ to be a quantum channel. In this case we call $\left(T(t)\right)_{t\ge0}$ a {\em{quantum dynamical semigroup}}. 

\subsection{Spectral Projections}
Consider an operator $M \in \cB(X)$ whose spectrum has a finite number of isolated points, $\lambda_j$, of magnitude $|\lambda_j| =1$, with $j \in \{1,2, \ldots, J\}$ for some $J \in \mathbb{N}$. Using the holomorphic functional calculus we can define the spectral projections corresponding to $\lambda_j$ by
	\begin{align}
	\label{eq:SpecProjector}
	P_j = \frac{1}{2\pi i}\oint_{\Gamma_j} \left(z-M\right)^{-1} dz,
	\end{align}
	where $\Gamma_j$ is any curve in $\CC$ enclosing only $\lambda_j$ but no other element of $\Spec(M)\backslash\{\lambda_j\}$. 
	
\comment{	Further, let $\tilde{P}_j$ denote the projection onto the eigenspace corresponding to the spectral point $\lambda_j$ which is defined through Yosida's Mean Ergodic Theorem for the rotated operator $\bar{\lambda_j}M$. See Section~\ref{ingredients}.}
	
	Note that in general $P_j$ will not be the projector onto the eigenspace $\operatorname{ker}(M-\lambda_j)$, since the quasi-nilpotent part 
	\begin{equation}
	\label{eq:Nilpotent}
	N_j := \left(\lambda_j - M\right)P_j  = \frac{1}{2\pi i}\oint_{\Gamma_j} (\lambda_j-z) \left(z- M\right)^{-1} dz ,
	\end{equation}
	is in general not equal to zero. In finite dimensions, $N_j$ is precisely the nilpotent part corresponding to the Jordan block of the eigenvalue $\lambda_j$. More precisely, since in finite dimensions the spectrum $\Spec(M)$ is finite and therefore discrete, we can write $M$ in its Jordan normal decomposition as 
	\begin{align*}
	M = \sum_{\lambda\in\Spec(M)}\lambda P_\lambda + N_\lambda.
	\end{align*}
	Here $P_\lambda$ and $N_\lambda$ are the spectral projectors and nilpotent parts corresponding to the eigenvalue $\lambda\in\Spec(M)$ defined analogously to \eqref{eq:SpecProjector} and \eqref{eq:Nilpotent}. In addition, in finite dimensions the spectral projectors and nilpotent parts satisfy
	\begin{align}
	P_\lambda P_\mu = \delta_{\mu\lambda}P_\mu,\quad
	N_\lambda P_\lambda = N_\lambda\\   
	N^{d_\lambda}_\lambda=0,\quad\text{with }d_\lambda = \tr(P_\lambda).
	\label{nil}
	\end{align}
	 Therefore, the question whether the nilpotent parts are zero or not is related to the diagonalizability of $M$.

	 In infinite dimensions $N_j$ is in general only quasi-nilpotent, i.e.~$\Spec(N_j) = \{0\}$. This can be seen by considering for any $\eps>0$ a closed curve $\Gamma_{j,\eps}$ with distance at most $\eps$ from $\lambda_j$ and not intersecting $\Spec(M)$. This yields
	for any $k\in\N$
	\begin{equation*}
	\left\|N_j^k\right\| \le \frac{1}{2\pi}\oint_{\Gamma_{j,\eps}} |\lambda_j-z|^k \left\|\left(z-M\right)^{-1}\right\| dz \le C_\eps \eps^k,
	\end{equation*}
	where $C_\eps>0$ is a constant dependent on $\eps$ but independent of $k$. Using the fact that $\eps>0$ is arbitrary we see that the spectral radius of $N_j$ vanishes:$$r(N_j)=\lim_{k\to\infty} \|N_j^k\|^{1/k} = 0,$$ and thus $\Spec(N_j) =\{0\}$.
	
	Under the assumption that the range of $P_j$ is finite-dimensional, which holds in particular if the underlying Banach space is finite dimensional, $M$ being a contraction implies $N_j=0$ \cite[Prop. 6.2]{Wolf12}. In infinite dimensions it is however possible to find contraction operators $M$ with only isolated spectral points on the unit circle and non-trivial quasi-nilpotent part as shown in the following example.
	\begin{example}
	\label{ex:VoltNeil}
		Let $V:L^2[0,1]\to L^2[0,1]$ be the Volterra operator 
		\[(Vf)(x):=\int_0^x f(t) \  dt \text{ with adjoint }(V^*f)(x)=\int_x^1 f(t) \ dt.\] 
		It is well-known that this operator has empty point spectrum $\Spec_p(V)=\emptyset$ and $\Spec(V)=\{0\}.$ This implies that $M:=(I+V)^{-1}$ exists and has spectrum $\Spec(M)=\{1\}$ and $\Spec_p(M)=\emptyset$ such that $\Vert M \Vert \ge 1.$
		On the other hand, 
		\[ \Vert M^{-1}f  \Vert^2 = \langle f+Vf,f+Vf \rangle = \Vert f \Vert^2 + 2\Re \langle Vf,f \rangle + \Vert V f \Vert^2 \ge \Vert f \Vert^2.\]
		Here, we used the fact that $2\Re \langle Vf,f \rangle = \langle f,V +V^* f\rangle \ge 0.$ To see this note that $Q:=V+V^*$ is a projection and therefore $Q\ge0$ which follows from 
		\begin{equation*}
		\begin{split}
		Q^2f = Q(Vf+V^*f) = Q  \int_0^1 f(t) \ dt = Qf = \int_0^1 f(t) \ dt.
		\end{split}
		\end{equation*}
		Hence, $\Vert M\Vert = 1$ due to $\|M\| = \sup_{f \neq 0}\frac{\|Mf\|}{\|f\|} = \sup_{f \neq 0}\frac{\|f\|}{\|M^{-1}f\|} \le \sup_{f \neq 0}\frac{\|f\|}{\|f\|} = 1.$
		Now, for $P$ being the spectral projector (defined through \eqref{eq:SpecProjector}) corresponding to the spectral point 1, this operator $M$ cannot satisfy $MP =PM = P.$ To see this, note that since 1 is the only point in the spectrum of $M$, we have that $P$ is equal to the identity of $L^2[0,1]$, and therefore $$MP = M \neq P = \1.$$
		Hence, the quasi-nilpotent operator  \eqref{eq:Nilpotent} corresponding to the isolated spectral point 1, i.e. $N=(1-M)P$, is not equal to zero.

	\end{example}
			
		\section{Main Results}\label{sec:main-results}	
In this section we state our main results on the quantum Zeno effect and the resulting quantum Zeno dynamics. 

\subsection{Case I: Uniformly continuous contraction semigroup}\label{sec:con-sg}
We start with the case in which the dynamics of the system is governed by a uniformly continuous contraction semigroup, and hence by a bounded generator. In the following, $X$ denotes a Banach space and $M \in \cB(X)$ denotes a contraction. Moreover, in  Theorem~\ref{thm:QuantitativQuantumZeno} (given below) we consider contractions $M$ which satisfy the following condition:
\smallskip

\noindent
\begin{ass}[Spectral gap assumption on $M$]
	$M$ has a finite number $J\in\N$ of points of modulus one in its spectrum such that the rest of the spectrum is contained in a disk of radius $0<\delta<1$, i.e.
	\begin{align} 
	\label{eq:SpectralGap0}
	\Spec(M) \subset B_\delta \cup \{\lambda_j\}_{j=1}^{J} \text{ with } |\lambda_j|=1,
	\end{align}
	where $B_{\delta}:=\{ z \in \CC; \vert z \vert \le \delta\}.$ An illustration of this spectral gap assumption is given in Figure~\ref{fig:SpectralGap}.
	\begin{figure}[t]
	 	\centering
	 		\includegraphics[width=0.55\linewidth]{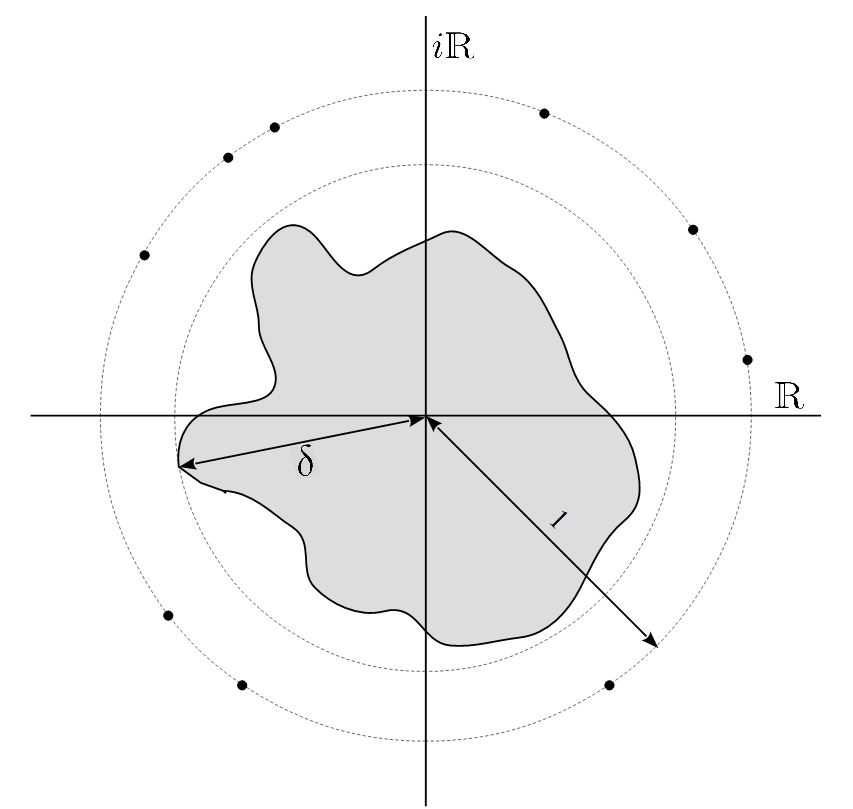}
	 	\caption{Illustration of the spectral gap assumption on $M$. The grey colored region as well as the dots on the unit circle form $\Spec(M).$}
	 	\label{fig:SpectralGap}
	 \end{figure}
	 \label{ass1}
\end{ass}
Our first main result is Theorem~\ref{thm:QuantitativQuantumZeno} stated below, which pertains to an open system whose evolution, governed by a contraction semigroup, is interrupted repeatedly and periodically by the action of a quantum operation $M$ which satisfies the following assumptions: 
\begin{enumerate}
\item the spectral gap assumption (Assumption~\ref{ass1}), and 
\item all the corresponding quasi-nilpotent parts are equal to zero. 
\end{enumerate} 
We establish a quantitative bound on the convergence rate of the Zeno product to the corresponding quantum Zeno dynamics.

Before stating the theorem, we would first like to discuss the assumptions above. In the setting of Theorem~\ref{thm:QuantitativQuantumZeno}, in order to study the quantum Zeno effect one needs to prove convergence (in operator norm) of the Zeno product $(Me^{t\cL/n})^n$ to the operator corresponding to the effective evolution within the quantum Zeno subspace. As a first step, one needs to find the condition on the spectrum of $M$ which would ensure such a convergence even in the trivial case in which $\cL=0$. This condition is precisely the spectral gap assumption (Assumption~\ref{ass1}) of $M$, along with the assumption that all the corresponding quasi-nilpotent parts are equal to zero. This is rigorously stated in Proposition~\ref{lem:SpectralGap} (see in particular point \ref{specNil}).


	\begin{theorem} \label{thm:QuantitativQuantumZeno}
	Let $\cL\in\cB(X)$ be a generator of a contraction semigroup and let $M\in\cB(X)$ be a contraction which satisfies the spectral gap assumption~\eqref{eq:SpectralGap0} with all corresponding quasi-nilpotent operators \eqref{eq:Nilpotent} being zero. Then for projections $P_j$ (defined through \eqref{eq:SpecProjector}) and any $n\in\mathbb{N}$,  $0<\delta<\tilde\delta<1$
	\begin{align}
	\label{eq:QuantitativeQuantumZeno}
	\left\|\left(Me^{t\cL/n}\right)^n - \sum_{j=1}^J e^{tP_j\cL P_j}\,\lambda_j^nP_j\right\| \le C\left(\frac{\|{\cL}\|}{n^{2/3}} + \frac{\|{\cL}\|^2}{n} + \tilde\delta^{n+1}\right),
	\end{align}
 with $C>0$ being a constant independent of $\cL$ and $n$.
\end{theorem}

Theorem~ \ref{thm:QuantitativQuantumZeno} is a quantitative version of a result by M\"obus and Wolf \cite[Theorem 1]{MobWolf19}. In \cite{MobWolf19} the role of quasi-nilpotent operators was not discussed and this extra assumption on $M$ for which the quasi-nilpotent parts vanishes, is missing. However, Example~\ref{ex:VoltNeil} shows that, unlike in finite dimensions, in infinite dimensions this assumption is not satisfied in general and Proposition~\ref{lem:SpectralGap} shows that it is necessary for the uniform convergence of the quantum Zeno limit. Furthermore, Theorem~\ref{thm:QuantitativQuantumZeno} extends existing results to more general operators $M$. This includes the quantum harmonic oscillator in Example~\ref{ex:HarmonicOsc}.
\begin{rem}
Theorem~\ref{thm:QuantitativQuantumZeno} states that frequent application of the quantum operation $M$ restricts the evolution of the system to the quantum Zeno subspace $\bigoplus_{j=1}^J\ran(P_j)$
	with the effective (i.e.~quantum Zeno) dynamics on each of the individual subspaces $\ran(P_j)$ given by $e^{tP_j\cL P_j}$. 
	We also note that the inequality \eqref{eq:QuantitativeQuantumZeno} can be alternatively stated as follows:
    \begin{equation}
    \begin{split}
    \left\|\left(Me^{t\cL/n}\right)^n - e^{t\sum_{j=1}^JP_j\cL P_j}M^n\right\| 
    &\le C\left(\frac{\|{\cL}\|}{\sqrt[3]{n^2}} + \frac{\|{\cL}\|^2}{n} + \tilde\delta^{n+1}\right)\\
    &=\mathcal O(n^{-2/3} \Vert \cL \Vert^2),
    \end{split}
    \end{equation}
    which closely resembles the form of the result in \cite{BuFaNaPaYu18} (compare Theorem 1 and Corollary 1 therein). This can be seen in the following way: Firstly, we note that a quantum operation $M$ which satisfies the spectral gap assumption \eqref{eq:SpectralGap}, and whose quasi-nilpotent operators are equal to zero, can be written as $M = \sum_{j=1}^J \lambda_j P_j + S,$
	where $S$ corresponds to the part of the spectrum with magnitude strictly smaller than 1, i.e.
	\begin{align*}
	S = \frac{1}{2\pi i}\oint_\gamma z \left(z-M\right)^{-1} dz,
	\end{align*}
	where $\gamma$ is a closed curve which encloses all parts of $\Spec(M)$ other than the isolated points $\lambda_j$ (for $j=1,2,\ldots, J$) on the unit circle.
	Using $P_jP_k = P_jS =S P_j = 0$ for all $j\neq k$, we get $M^n = \sum_{j=1}^J \lambda_j^n P_j +S^n$
	and $\|S^n\|\le C\tilde\delta^{n+1}$. Hence, 
	\begin{align*}
	e^{t\sum_{j=1}^JP_j\cL P_j}M^n &= \sum_{k=1}^J e^{t\sum_{j=1}^JP_j\cL P_j} \,\lambda_k^nP_k + \mathcal{O}(\tilde\delta^{n+1}) \\&=\sum_{k=1}^J\, \prod_{j=1}^Je^{tP_j\cL P_j}\,\lambda_k^nP_k + \mathcal{O}(\tilde\delta^{n+1})\\ &=\sum_{k=1}^J e^{tP_k\cL P_k}\,\lambda_k^nP_k + \mathcal{O}(\tilde\delta^{n+1}).
	\end{align*}
	In the second line, we have used the fact that the operators $P_j\cL P_j$ for $j \in \{1,2,\ldots, J\}$ commute with each other. In the third line we have used the fact that  for each fixed $k$ in the sum and $j\neq k$, the only term in the series expansion of the exponential $e^{tP_j\cL P_j}$ which makes a non-trivial contribution to the sum is the zeroth order term.
\end{rem} 

Our next result, given by Theorem~\ref{thm:StrongZenoGeneral} below, shows convergence of the Zeno product under weaker assumptions than the ones used in Theorem~\ref{thm:QuantitativQuantumZeno}. It establishes a novel strong quantum Zeno limit for quantum operations which do not satisfy the spectral gap assumption (Assumption~\ref{ass1}) but instead satisfy a weaker property of strong power-convergence (see (\ref{eq:StrongPowerConv}) below). 
More precisely, we prove strong convergence of the Zeno product for contractions $M$ which are strongly power-convergent to the projection onto the corresponding invariant subspace. To our knowledge this is first result on the quantum Zeno effect for general quantum operations which does not rely on a spectral gap assumption, 
and it applies to the (bosonic quantum-limited) attenuator channel, discussed in Example~\ref{ex:Atten}, which is an important example of a quantum channel arising in quantum optics.
\begin{theorem}
	\label{thm:StrongZenoGeneral}
Let $\cL\in\cB(X)$ be a generator of a contraction semigroup and $M\in\cB(X)$ a contraction which satisfies for all $x\in X$
\begin{align}
\label{eq:StrongPowerConv}
\lim_{n\to\infty}M^n x =Px 
\end{align}
for some operator $P\in\cB(X)$. Then
\begin{align}
\lim_{n\to\infty}\left(Me^{t\cL/n}\right)^n x = e^{tP\cL P}Px
\end{align}
for all $x\in X$.
\end{theorem}
Note that~\eqref{eq:StrongPowerConv} implies that the operator $P$ is the projection onto the invariant subspace of $M$.
\begin{rem}
For the special case that $X$ is the space of trace-class operators over some Hilbert space, it is known \cite{A81} that $\| M^n x -Px \|_1 \xrightarrow[n \to \infty]{} 0$ if and only if $\| M^n x \|_1 \rightarrow \| Px \|_1$ and $M^n x$ is weakly convergent to $Px.$ Therefore, often (e.g.~when $M$ is a quantum channel) it is enough to just assume a \emph{weak power-convergence} in the above theorem.
\end{rem}

\subsection{Case II: Strongly continuous semigroups}
Our third result, stated in Theorem~\ref{thm:UnboundedQuantumZenoProjector}, pertains to open systems whose evolution is governed by a strongly continuous quantum dynamical semigroup (and hence by an unbounded generator). Once again the evolution is interrupted by repeated and periodic actions of a quantum operation $M$ satisfying the assumptions $(i)$ and $(ii)$ stated in Section~\ref{sec:con-sg}. In this case we obtain a bound on the speed of convergence to the quantum Zeno dynamics in the strong topology.

\begin{theorem}
	\label{thm:UnboundedQuantumZenoProjector}
	Let $\cL$ with domain $\cD(\cL)$ be a generator of a strongly continuous contraction semigroup $\left(e^{t\cL}\right)_{t\ge0}$ and $M\in\cB(X)$ be a contraction satisfying 
	the spectral gap assumption~\eqref{eq:SpectralGap0} with all corresponding quasi-nilpotent operators (defined through~\eqref{eq:Nilpotent}) being zero. Moreover, assume that $M\cL, \cL M$, are both densely defined and bounded. Then for all $x\in\cD(\cL)$, $n\in\mathbb{N}$, and $0<\delta<\tilde\delta<1$
	\begin{equation}
	\begin{split}
	\label{eq:UnboundedAsymptotic}
	\left\|\left(\left(Me^{t\cL/n}\right)^n - \sum_{j=1}^J e^{tP_j\cL P_j}\lambda^n_jP_j\right)x\right\|&\le C\Bigg(\left(\frac{1}{\sqrt[3]{n}} + \tilde\delta^{n+1}\right)\|x\| + \frac{\|\cL x\|}{\sqrt[3]{n^4}}\Bigg)\\
	&=\mathcal O(\Vert x \Vert_{\mathcal D(\mathcal L)} n^{-1/3}),
	\end{split}
	\end{equation}
    where $\|\cdot\|_{\cD(\cL)}$ denotes the graph norm, i.e. $\|x\|_{\cD(\cL)} = \|x\| + \|\cL x\|$. Consequently, we have for all $x \in X$ 
	\begin{equation}
	\label{eq:UnboundedLimit}
	\left\|\left(\left(Me^{t\cL/n}\right)^n - \sum_{j=1}^J e^{tP_j\cL P_j}\lambda^n_jP_j\right)x\right\|\xrightarrow[n\to\infty]{} 0. 
\end{equation}
\end{theorem}
It is important to determine the rate, topology and set of states for which the quantum Zeno product $(Me^{t\cL/n})^n$ converges to the quantum Zeno dynamics. We recall that we refer to the asymptotic of
$(Me^{t/n
\cL})^n$
in a certain topology, as the {\em{quantum Zeno limit}}. 
In the setting of Theorem~\ref{thm:QuantitativQuantumZeno}, the limit is in the uniform topology, whereas in the setting of Theorem~\ref{thm:StrongZenoGeneral} and Theorem~\ref{thm:UnboundedQuantumZenoProjector} the limit is in the strong topology.

\smallskip
As mentioned earlier, Proposition~\ref{lem:SpectralGap}, given below, shows the requirement of the spectral gap condition (Assumption~\ref{ass1}) on the quantum operation $M$, to obtain a quantum Zeno limit in operator norm. In the trivial case in which there is no additional quantum dynamics, i.e.~$\mathcal L=0$, the quantum Zeno product reduces simply to $\left(Me^{t\cL/n}\right)^n=M^n.$ \comment{This expression, if the quantum Zeno effect holds for the channel $M$, is then supposed to approximate \ro{$\sum_{j=1}^JM^n\tilde P_j$}, where $\tilde P_j$ is the projection onto the eigenvalue $\lambda_j$ on the unit circle with $j=1,..,J$. In particular, if the quasi-nilpotent part \eqref{eq:Nilpotent} associated with eigenvalues $\lambda_j$ is zero, it follows that 
\[  \sum_{j=1}^J M^n\tilde P_j = \sum_{j=1}^J\lambda_j^n \tilde P_j. \]
Thus, if we are concerned about the existence of the quantum Zeno limit in operator norm, we are left to study the existence of the limit
\[\lim_{n \rightarrow \infty}\left\Vert \left(Me^{t\cL/n}\right)^n - \sum_{j=1}^J M^n\tilde P_j \right\Vert=\lim_{n \rightarrow \infty} \left\Vert M^n -\sum_{j=1}^J\lambda_j^n \tilde P_j\right\Vert=0, \]
where the first identity holds because of the choice $\cL=0$.
The following proposition shows that the existence of this limit, stated as case (\ref{uniformconv}), in operator norm is equivalent to exponentially fast convergence of $\left\Vert M^n -\sum_{j=1}\lambda_j^n \tilde P_j \right\Vert$ to zero (\ref{uniformconvExpDec}), and finally that $M$ satisfies a spectral gap condition \eqref{three} which will be assumed throughout this article.}

\noindent
\begin{prop}
		\label{lem:SpectralGap}
		Let $M\in\cB(X)$ be a contraction, $J\in\N$, $\{\lambda_j\}_{j=1}^J\subset\C$ with $|\lambda_j|=1$. Then the following are equivalent:
		\begin{enumerate}
		\item $\lim_{n\to\infty}\| M^n - \sum_{j=1}^J \lambda_j^n  K_j\| = 0,$ for some $0\neq K_j\in\cB(X)$.\label{uniformconv}
		\item $\| M^n - \sum_{j=1}^J \lambda_j^n K_j\| \le C  \tilde\delta^{n+1}$ for some $0\neq K_j\in\cB(X)$, $0\le\tilde\delta<1$ and $C>0.$ \label{uniformconvExpDec}
		\item \label{three} For some $0\le\delta<1$ the contraction $M$ satisfies the spectral gap condition given by
		\begin{align}
		\label{eq:SpectralGap}
		\{\lambda_j\}_{j=1,\dots,J}\subset\Spec(M) \subset B_\delta \cup \{\lambda_j\}_{j=1}^J,
		\end{align}
		where $B_{\delta}:=\{ z \in \CC; \vert z \vert \le \delta\},$ and its quasi-nilpotent parts $N_j$ (defined through  \eqref{eq:Nilpotent}) are equal to zero
		for all $j =1,\cdots, J$. \label{specNil}
		\end{enumerate}
		If either of the above condition holds, then the spectral projectors $P_j$ (defined through~\eqref{eq:SpecProjector}) are well-defined, we have $K_j=P_j$ and each $P_j$ is the projector onto the eigenspace corresponding to the eigenvalue $\lambda_j$.
	\end{prop}

In particular from Proposition~\ref{lem:SpectralGap} we can immediately infer the following corollary, which shows the equivalence of the uniform convergence of the powers $M^n$ to a spectral gap condition on $M$.
\begin{corr}
\label{cor:uniformpowerconv}
Let $M\in\cB(X)$ be a contraction. Then the following are equivalent:
\begin{enumerate}
\item $\left(M^n\right)_{n\in\N}$ converges uniformly.
\item $\left(M^n\right)_{n\in\N}$ converges uniformly with exponential convergence rate.
\item For some $0\le\delta<1$ the contraction $M$ satisfies the spectral gap condition given by
	\begin{align}
		\Spec(M)\subset B_\delta \cup \{1\},
\end{align}
where $B_{\delta}:=\{ z \in \CC; \vert z \vert \le \delta\},$ and in the case $1\in\Spec(M)$ that the corresponding quasi-nilpotent part is equal to zero.
\end{enumerate}
\end{corr}

We now give two examples of quantum channels $M$ which satisfy the condition \eqref{specNil} in Proposition~\ref{lem:SpectralGap} and hence the assumption in Theorem~\ref{thm:QuantitativQuantumZeno}.
\begin{example}[Generalised depolarising channel] 
\label{ex:GDC}
	Consider for $X = \cT(\cH)$, $\sigma\in\cT(\cH)$ and $p\in[0,1)$ the contraction $M$ being the generalised depolarising channel $\Phi_p$,
	which acts on any state $\rho \in \cT(\cH)$ as follows:
\[\Phi_p(\rho) = (1-p)\rho + p\tr(\rho)\sigma.\]
	We can directly construct the resolvent for any complex number $\lambda \notin \{1-p,1\}$ by
	\begin{align}
	\label{eq:PolResolvevnt}
	\left(\lambda - \Phi_p\right)^{-1}(\rho) =  \frac{\rho + \frac{p}{\lambda-1}\tr(\rho)\sigma}{\lambda+ p -1},
	\end{align}
	for any $\rho\in\cT(\cH)$, which shows that $\Spec(\Phi_p) \subset \{(1-p),1\}$. Moreover, using the explicit form of the resolvent (given in~\eqref{eq:PolResolvevnt}) we can directly compute the projector $P$ corresponding to the spectral point $1$ (cf~\eqref{eq:SpecProjector}):
	\begin{align*}
	P(\rho) &= \frac{1}{2\pi i}\oint_\Gamma \left(z - \Phi_p\right)^{-1}(\rho) dz = \frac{1}{2\pi i}\oint_\Gamma \frac{\rho+ \frac{p}{z-1}\tr(\rho)\sigma}{z+ p -1} dz \\ &=\frac{1}{2\pi i}\oint_\Gamma\frac{p}{(z-1)(z+ p -1)}dz\,\tr(\rho)\sigma = \tr(\rho)\sigma,
	\end{align*}
	where $\Gamma$ encloses the spectral point $1$ but not $1-p$.
	Note that the spectral projection $ P$ coincides with the projection onto the invariant subspace  $\F = \operatorname{span}\{\sigma\}$. Hence, we have explicitly shown that the quasi-nilpotent part of the generalised depolarising channel is zero and hence assumption \eqref{specNil} in Proposition~\ref{lem:SpectralGap} and the assumption of Theorem~\ref{thm:QuantitativQuantumZeno} holds. Moreover, as $P \neq \1$ we have also shown the equality $\Spec(\Phi_p) = \{(1-p),1\}$.
\end{example}

\begin{example}[Schrödinger evolution of the harmonic oscillator] \label{ex:HarmonicOsc}
	We consider the Hamiltonian of a one-dimensional quantum harmonic oscillator $H = -\Delta + \omega^2 x^2$ defining a strongly continuous group $(U(t))_{t \in {\mathbb{R}}}$ on $L^2(\R)$, where $U(t) = e^{-itH}$. Let $E_n := \omega\left(n+1/2\right)$ denote the energy eigenvalues and let $\{\ket{n}\}_{n \in \mathbb{N}}$ denote the energy eigenbasis of $H$. Then $U(t) = \sum_{n=0}^\infty e^{-itE_n} |n\rangle\langle n|$, where the series converges strongly in $L^2(\R)$. Consider  corresponding quantum channel $\Phi_{U(t)}$ on $\cT(L^2(\R))$ given by conjugating with $U(t)$, i.e.
	\begin{align*}
	\Phi_{U(t)}(\rho) &= U(t)\rho U(t)^* = \sum_{n,m=0}^\infty e^{-it(E_n-E_m)}\langle n|\rho|m\rangle\ket{n}\bra{m} \\ &= \sum_{n,m=0}^\infty e^{-it\omega(n-m)}\langle n|\rho|m\rangle\ket{n}\bra{m},
	\end{align*}
	where the convergence of the series is in trace norm. We now see that for all $\lambda\notin\overline{\{e^{-it\omega k}\}}_{k\in\mathbbm{Z}}$ we can explicitly write down the resolvent of $\Phi_{U(t)}$ at $\lambda$, which is 
	\begin{align*}
	\left(\lambda - \Phi_{U(t)}\right)^{-1}(\rho) = \sum_{n,m=0}^\infty \frac{\langle n|\rho|m\rangle\ket{n}\bra{m}
}{\lambda -e^{-it\omega(n-m)}}	\end{align*}
	and hence $\Spec(\Phi_{U(t)}) \subset \overline{\{e^{-it\omega k}\}}_{k\in\mathbbm{Z}}$.
	
	Now consider a fixed time $t$ satisfying $t\omega = 2\pi/k$ for some $k\in\N$ and define the contraction $M:=\Phi_{U(t)}$. In that case we see that $\Spec(\Phi_{U(t)})$ consists at most of $k$ points, which are hence all isolated and therefore $\Phi_{U(t)}$ satisfies \eqref{eq:SpectralGap}. For any $j = 0,\cdots,k$ let $\Gamma_j$ be a closed curve surrounding the spectral point $\lambda_j=e^{-\frac{2\pi i j}{k}}$ and separating this point from the rest of $\Spec(\Phi_{U(t)})$. We can then compute the spectral projector corresponding to $\lambda_j$ 
	which is
	\begin{align*}
	P_j(\rho) = \frac{1}{2\pi i }\oint_{\Gamma_j} \left(z-\Phi_{U(t)}\right)^{-1}(\rho) dz 
	&=  \sum_{n,m=0}^\infty\frac{1}{2\pi i }\oint_{\Gamma_j}\frac{\,\langle n|\rho|m\rangle\ket{n}\bra{m} }{z -e^{-\frac{2\pi i(n-m)}{k}}}dz \\
	&=  \sum_{\substack{n,m=0\\n-m= j \text{mod}\, k}}^\infty \langle n|\rho|m\rangle\ket{n}\bra{ m}.
	\end{align*}
	Hence, as for all $\rho\in\cT(L^2(\R))$ the image of the spectral projector $P_j(\rho)$, if non-zero, is an eigenvector corresponding to the eigenvalue $\lambda_j = e^{-\frac{2\pi i j}{k}}$ of $\Phi_{U(t)}$. Thus, all quasi-nilpotent parts are equal to zero, which shows that $M=\Phi_{U(t)}$ fulfills the condition \eqref{specNil} in Proposition~\ref{lem:SpectralGap}. 
\end{example}
In the following example we see that the (bosonic quantum-limited) attenuator channel does not satisfy the spectral gap assumption used in Theorem~\ref{thm:QuantitativQuantumZeno}. However, it is still strongly power-convergent to its invariant subspace, i.e.~it satisfies the condition \eqref{eq:StrongPowerConv}. Hence, Theorem~\ref{thm:StrongZenoGeneral} applies for the choice of $M$ being the attenuator channel.
\begin{example}[Attenuator Channel]
\label{ex:Atten}
Let $\Phi^{att}_t$ be the attenuator channel with attenuation parameter $\eta(t) = e^{-t}$ which acts on an arbitrary state $\rho$ as
\begin{align*}
\Phi^{att}_t(\rho) = \sum_{l=0}^\infty \frac{(1-e^{-t})}{l!} e^{-tN/2}a^l\rho (a ^*)^le^{-tN/2} = \sum_{l=0}^\infty K_l(t)\rho K^*_l(t),
\end{align*}
with 
\begin{align*}
K_l(t) =  \frac{(1-e^{-t})}{l!}e^{-tN/2}a^l = \sum_{m=0}^\infty \sqrt{\binom{m+l}{m}}(1-e^{-t})^{l/2}e^{-tm/2}\ket{m}\bra{m+l} 
\end{align*}
(cf. \cite[Lemma II.12]{DTG16}).
From the above one can see that the  attenuator channel has a unique invariant state given by $\kb{0}$ and converges strongly to the projector of this invariant state in the limit $t\to\infty$, i.e. for all states $\rho$
\begin{align}
\label{eq:attconv}
\lim_{t\to\infty} \Phi^{att}_t(\rho) =\tr(\rho)\kb{0} = P(\rho),
\end{align}
where we defined the projector $P(\cdot) = \tr(\cdot)\kb{0}$.
Consider the quantum operation $$M=\Phi^{att}_{t_0},$$
where $t_0>0$ is any fixed time. 
Using the fact that $\left(\Phi^{att}_t\right)_{t\ge 0}$ is a semigroup,
\eqref{eq:attconv} immediately gives that $M$ is strongly power-convergent, i.e.
\begin{align*}
\lim_{n\to\infty}M^n(\rho) = P(\rho). 
\end{align*}
However, $M$ is not uniformly power-convergent. This can be seen by using the fact that for coherent states, $$\ket{\alpha}= e^{-|\alpha|^2/2} \sum_{m=0}^\infty\frac{\alpha^m}{m!}\ket{m},$$ the attenuator channel acts as
\begin{align*}
\Phi^{att}_t(\kb{\alpha}) = \kb{e^{-t}\alpha}.
\end{align*}
Hence, for all $n\in\N$
\begin{align*}
\|M^n-P\|&=\sup_{\|x\|_1=1}\left\|M^n(x) - P(x)\right\|_1 \ge \sup_{\kb{\alpha}} \left\|M^n(\kb{\alpha}) - P(\kb{\alpha})\right\|_1 \\&= \sup_{\kb{\alpha}} \left\| \kb{e^{-nt_0}\alpha} - \kb{0}\right\|_1 = 2.
\end{align*}
Hence, we see that $M$ violates the assumptions in Corollary~\ref{cor:uniformpowerconv} and Theorem~\ref{thm:QuantitativQuantumZeno}, i.e. either 1 is not an isolated point in the spectrum of $M$ or its associated quasi-nilpotent part is not equal to zero. 
\end{example}
In Section~\ref{sec:strongpowconv} we show  strong power-convergence for a variety of other quantum channels, which provides more examples to which our Theorem~\ref{thm:StrongZenoGeneral} can be applied. These include quantum channels related to the quantum Ornstein-Uhlenbeck semigroup (Example~\ref{ex:qOU}), the Jaynes-Cummings model (Example~\ref{ex:Jaynes-Cumming}) and photon absorption and emission processes (Example~\ref{ex:EmAbs} and \ref{ex:photonabs}). To prove that these examples of quantum channels satisfy the strong power-convergence property required in Theorem~\ref{thm:StrongZenoGeneral}, we use an embedding technique into the Hilbert space of Hilbert-Schmidt operators, developed in \cite{CF00}, and the results on ergodic theory of quantum Markov semigroups in~\cite{FV82,DFR10}.

\medskip

For the following example, we investigated numerically the speed of convergence towards the Zeno subspace and compared it to our analytical bound \eqref{eq:QuantitativeQuantumZeno}.
\begin{figure}[t]
	\centering
		\includegraphics[width=0.65\linewidth]{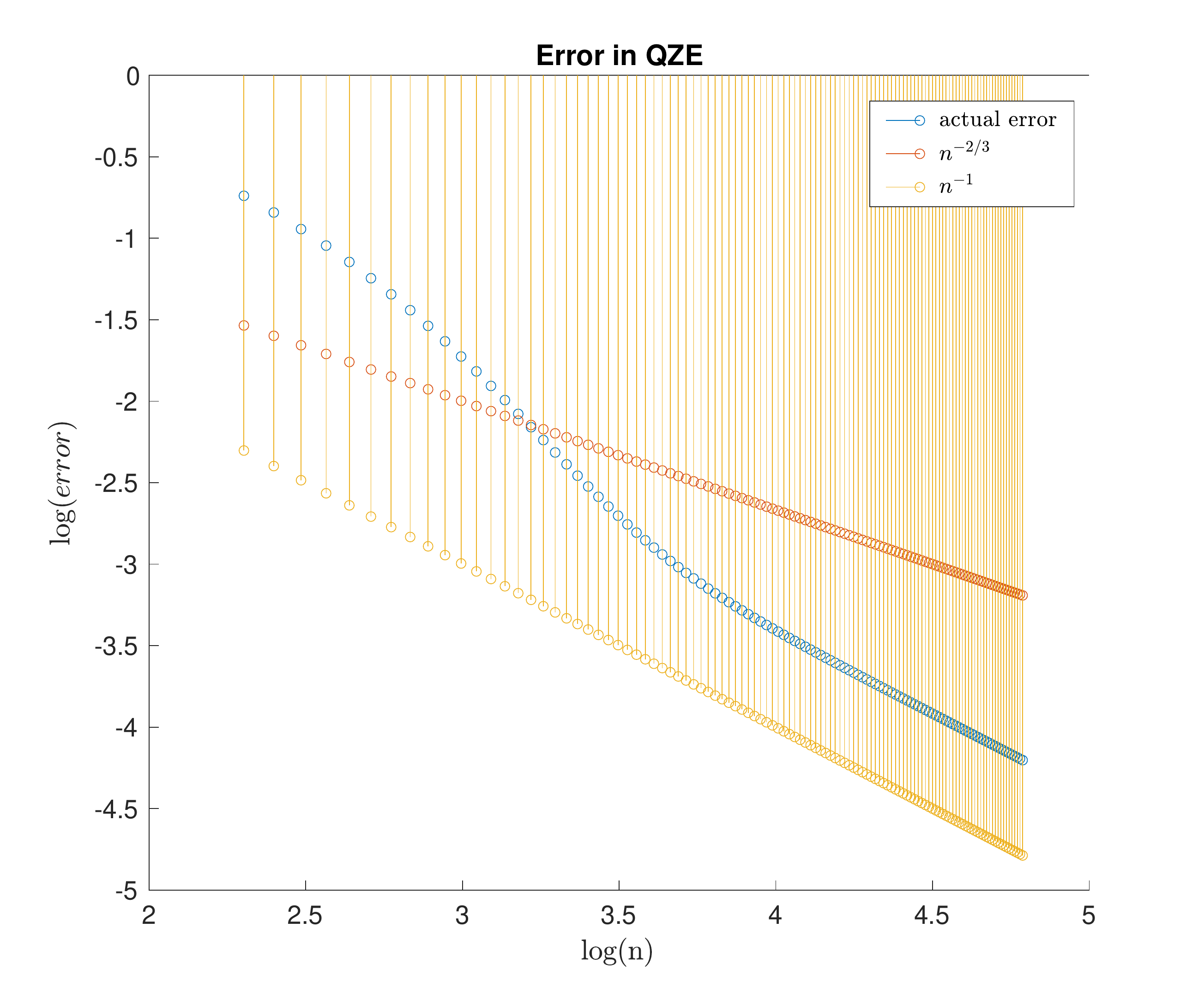}
		\caption{We compare the numerically computed error in the quantum Zeno limit and find a speed of convergence $\propto n^{-1}$. This is to be compared with the analytically obtained decay rate $\propto n^{-2/3}$ predicted by \eqref{eq:QuantitativeQuantumZeno}. }
		\label{fig:error}
\end{figure}

	\begin{example}
	    We consider for $M$ the generalized depolarising channel, introduced in Example \ref{ex:GDC}, with 
	   \begin{equation*}
	    \sigma=\frac{1}{3}\left(\ket{0}\bra{0}+\ket{1}\bra{1}+\ket{2}\bra{2}\right)+ \frac{1}{10}\left( \ket{0}\bra{1}+\ket{1}\bra{0}\right),
	    \end{equation*} The dynamics is given by $U(t)\rho U(t)^\dagger=e^{-iHt} \rho e^{iHt}$ where $H=-\Delta+ x^2$ is the Hamiltonian of the harmonic oscillator. For an initial state $\rho=\ket{0}\bra{0},$ analysing the quantum Zeno limit reduces to studying the norm $\left\|\left(Me^{t\cL/n}\right)^n(\rho) - \sigma \right\|_1.$ The analytical error ($\propto n^{-3/2}$) and numerical error ($\propto n^{-1}$) are both illustrated in Figure \ref{fig:error}.
	\end{example}

The following example shows that the boundedness assumption in Theorem \ref{thm:UnboundedQuantumZenoProjector} is strictly necessary. Moreover, in this example both the pointwise quantum Zeno limit as well as the quantum Zeno dynamics do not exist, if this assumption is not satisfied.
\begin{example}
Consider the state 
\[ \ket{\varphi} = \sum_{n=1}^{\infty} 2^{-n/2} \ket{2^n}\]
and $P=\kb{\varphi}$ be the projection onto that state. Let $\mathcal L:=iN$ with $N$ being the number operator. Then the $k^{th}$ Yosida approximant $\cL_k$ acts on the state $\ket{\varphi}$ as follows:
\[ \cL_k \ket{\varphi} \equiv k\cL \left(k-\cL\right)^{-1} \ket{\varphi} = \sum_{n=1}^{\infty} i2^{n/2} k(k-i2^{n})^{-1} \ket{2^n}.\]
This implies that 
\[ P\cL_k P = \ket{\varphi}\bra{\varphi}  \sum_{n=1}^{\infty} \frac{ik}{k-i2^{n}}\]
and hence 
\[ e^{t P\cL_k P} = e^{ t\sum_{n=1}^{\infty} \frac{ik}{k-i2^{n}}} P.\]
Thus, if $x \in \operatorname{span}(\varphi)^{\perp}$ it follows that 
\begin{align}
 e^{t P\cL_k P}x =0 \quad \forall k, \label{ortho}   
\end{align} 
whereas for $x \notin \operatorname{span}(\varphi)^{\perp}$, the limit as $k$ tends to infinity, of the left hand side of the above equation, does not exist. This shows the non-existence of the limit of the Yosida approximation of the Zeno dynamics given $e^{tP\cL_k P}$.

Turning now to the Zeno product, we start by observing that,
\[ \bra{\varphi} e^{t/n \mathcal L} \ket{\varphi} = \sum_{k=1}^{\infty} 2^{-k}e^{it 2^{k}/n}\]
so that
\[\left( Pe^{t/n \mathcal L}\right)^n \ket{\varphi}  = \ket{\varphi} \left(\sum_{k=1}^{\infty} 2^{-k}e^{it 2^{k}/n}\right)^n.\]
An elementary calculation (that we leave to the reader) shows that the limit of the Zeno product $\left( Pe^{t/n \mathcal L}\right)^n$ as $n$ tends to infinity also does not exist.
\end{example}


	\section{Proof of Proposition~\ref{lem:SpectralGap}}\label{sec;Prop-examples}
	We start this section by first introducing certain elements of ergodic theory which we employ as ingredients of the proof.
	The invariant subspace of a contraction $M \in \cB(X)$ shall be denoted by $\F:=\Big\{x\in X\Big| Mx = x\Big\}.$ Consider for $n\in\N$ the average operator
	\begin{align}
	\label{eq:AverageOp}
	A_n := \frac{1}{n}\sum_{k=0}^{n-1}M^k.
	\end{align}
	The {\em{mean ergodic subspace}} of $M$, which is the subspace of $X$ on which \eqref{eq:AverageOp} has a strong limit, shall be denoted by
\[	X_{\text{me}} = \Big\{x\in X\Big| \lim_{n\to\infty}A_n x \text{ exists}\Big\}.\]

	{\em{Yosida's Mean Ergodic Theorem}} (cf. \cite[Chapter 2]{Krengel85} or \cite[Chapter VIII. 3.]{Yos80}) gives the following complete characterisation of the mean ergodic subspace,
	\begin{align*} 
	X_{\text{me}} = \F \oplus\overline{\operatorname{ran}\left(\1 -M\right)},
	\end{align*}
	and in addition states that for all $x\in X_{\text{me}}$ the average operator converges to some operator $\tilde P$
	\begin{equation}
	\label{eq:tildeP}
	\lim_{n\to\infty}A_nx = \tilde Px\quad \text{ with }\quad \tilde P^2 = \tilde P,
	\end{equation} 
	defined on the subspace $X_{\text{me}}$. Here $\tilde P$ is the projection onto the invariant subspace $\F$, i.e.~$\operatorname{ran}(\tilde P) = \F, \ \operatorname{ker}(\tilde P) = \overline{\left(\1 -M\right)X}$ and $\tilde P M = M\tilde P = \tilde P$. We call the operator $M$ \emph{mean ergodic}, if $X = X_{\text{me}}.$  
	
With these preliminaries in hand, we are now ready to state the proof of Proposiion~\ref{lem:SpectralGap}:
	
\begin{proof}[Proof of Prop. \ref{lem:SpectralGap}]
The direction  $\eqref{uniformconvExpDec}\implies\eqref{uniformconv}$ is trivial. We continue by showing the implication $\eqref{specNil}\implies\eqref{uniformconvExpDec}$. As the quasi-nilpotent parts are all equal to zero, we can pick $K_j=P_j$, with spectral projector $P_j$ as defined in \eqref{eq:SpecProjector}, and get for $0\le\delta<\tilde\delta<1$ the estimate for the expression in \eqref{uniformconvExpDec}:
\[\left\|M^n - \sum_{j=1}^J \lambda_j^n P_j\right\| = \left\|\frac{1}{2\pi i}\oint_{\partial B_{\tilde\delta}} z^n \left(z-M\right)^{-1} dz\right\|\le C \tilde\delta^{n+1},\]
where we used the fact that $\|\left(z-M\right)^{-1}\|$ is uniformly bounded for $z\in\partial B_{\tilde\delta}$. 

We complete the proof by showing $\eqref{uniformconv} \implies\eqref{specNil}$. Note that defining for each $\lambda_j$ the rotated operator $M_j = \overline{\lambda_j}M$ and considering the corresponding average operator
\begin{equation}
\label{eq:RotatedAverage}
A_n(M_j):=\frac{1}{n}\sum_{k=0}^{n-1} M_j^k, 
\end{equation}
\eqref{uniformconv} implies that $A_n(M_j)$ converges uniformly to $K_j$ and hence $M_j$ is uniformly mean ergodic. Hence, 
Yosida's mean ergodic theorem implies that $K_j=\tilde P_j$, with $\tilde P_j$ being the projector onto the invariant subspace of $M_j$.
Now using~\cite[Theorem \@1]{Llyod81} and the arguments therein we see that the restriction $M_j|_{\operatorname{ran}(\1- \tilde P_j)}$ does not contain $1$ in its spectrum. Since $\Spec(M_j) = \Spec\left(M_j|_{\operatorname{ran}(\1- \tilde P_j)}\right) \cup \Spec\left(M_j|_{\operatorname{ran}(\tilde P_j)}\right)$ and $\Spec\left(M|_{\operatorname{ran}(\tilde P_j)}\right) \subset \{1\}$, we see $1\notin\overline{\Spec(M_j)\setminus\{1\}}\subset \Spec\left(M|_{\operatorname{ran}(\1- \tilde P)}\right)$, which shows that $1$ is an isolated point in $\Spec(M_j)$. By rotating $M_j$ back to $M$, we see that each $\lambda_j$ is an isolated point in $\Spec(M)$. Moreover, again using~\cite[Theorem \@1]{Llyod81}, we see that all poles of the resolvent are of first order at each of the isolated spectral points $\lambda_j$, i.e.~$(z -\lambda_j)\|\left(z-M\right)^{-1}\|$ is bounded in $z$. Now consider for an arbitrary $\eps>0$ a closed curve $\Gamma_{j,\eps} \subset \CC$ with distance at most $\eps$ to $\lambda_j$ and not intersecting $\Spec(M)$. Then we can bound
\[\|N_j\| \le \frac{1}{2\pi} \oint_{\Gamma_{j,\eps}} |\lambda_j - z| \left\|\left(z-M\right)^{-1}\right\| dz \le C \,\eps\]
for some $C>0$ independent of $\eps$. Since $\eps>0$ is arbitrary, we see that all quasi-nilpotent operators $N_j=0$ and hence the spectral projectors \eqref{eq:SpecProjector} are equal to the projections onto the corresponding eigenspaces $P_j = \tilde P_j =K_j$.

 In order to conclude the spectral gap condition \eqref{eq:SpectralGap} and hence complete the proof, it suffices to show that apart from the $\lambda_j$ there are no other points lying in the intersection of $\Spec(M)$ and the unit circle in the complex plane. To show this, let $\gamma$ be a closed curve in the complex plane enclosing $\Spec(M)\setminus\{\lambda_j\}_{j=1}^J$ and separating it from $\{\lambda_j\}_{j=1}^J$. Then, 
 \begin{align*}
 Q = \frac{1}{2\pi i}\oint_\gamma \left(z-M\right)^{-1} dz
 \end{align*}   
is the spectral projector corresponding to $\Spec(M)\setminus\{\lambda_j\}_{j=1}^J$. From \cite[Theorem 6.17, Chapter III §6.4.]{Kato} it follows that $\Spec(MQ) = \Spec(M)\setminus\{\lambda_j\}_{j=1}^J$.  Moreover, since $\tilde P_j Q = P_j Q= 0$ for all $j$, we can conclude from~\eqref{uniformconv} that $\lim_{n\to\infty}\|(MQ)^n\| = 0.$
 By the spectral mapping theorem applied to polynomials (see e.g.~\cite[Corollary 1, Chapter VIII 7]{Yos80}) we have $$\Spec((MQ)^n) =(\Spec(MQ))^n =(\Spec(M)\setminus\{\lambda_j\}_{j=1}^J)^n $$ for each $n\in\N$. Hence, we see that there is no point in $\Spec(M)\setminus\{\lambda_j\}_{j=1}^J$ on the unit circle, which gives the spectral gap condition \eqref{eq:SpectralGap} and completes the proof.
\end{proof}

\section{Proof of Theorem \ref{thm:QuantitativQuantumZeno}} \label{sec:QuantitativeBoundedZeno}
In this section we give the proof of Theorem~\ref{thm:QuantitativQuantumZeno}.
 As in the previous section, we consider closed curves $\Gamma_j$ enclosing the isolated spectral points $\lambda_j$ of the quantum operation $M$ on the unit circle and separating them from $\Spec(M)\backslash \{\lambda_j\}$. We choose the $\Gamma_j$ in such a way that their distance from $\lambda_j$ is small, say at most $1/2$. Moreover, for $0<\delta<\tilde\delta<1$ we consider a closed curve $\gamma \subset B_{\tilde\delta}$ that satisfies $\gamma \cap B_\delta=\emptyset$. For convenience we denote by $\Omega$ the open set which lies in the interior of all of the curves $\Gamma_j$ and $\gamma$. For an illustration of this construction see Figure~\ref{fig:SpecCurvy}. 
\medskip

\begin{figure}[t]
	\centering
		\includegraphics[width=0.55\linewidth]{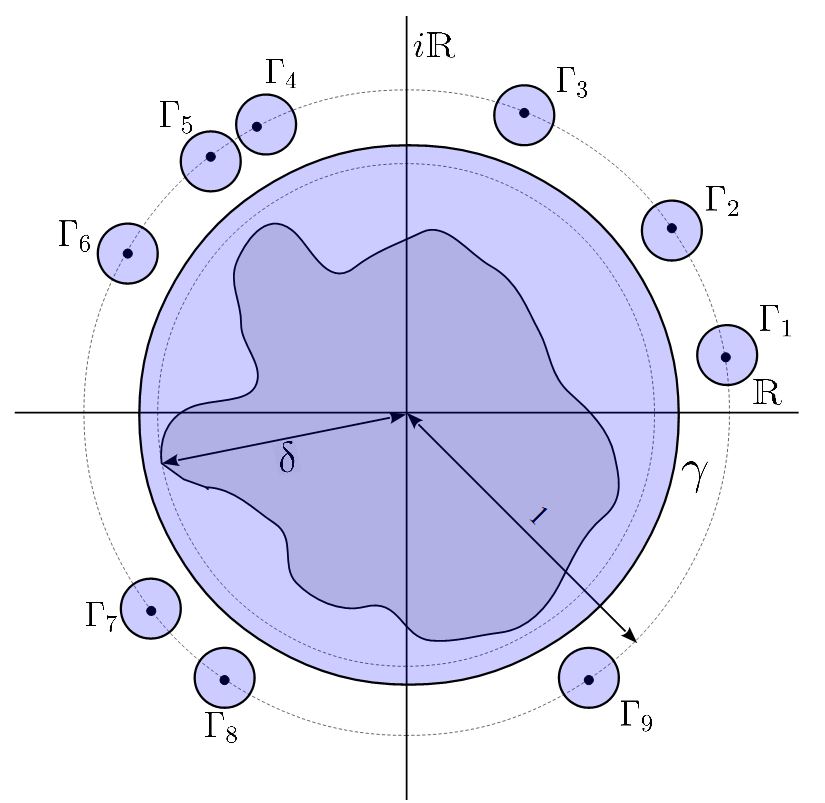}
		\caption{$\Spec(M)$ with curves $\Gamma_j$ and $\gamma$. The spectrum of $M$ consists of the dark region in the middle with maximal distance from the origin equal to $\delta$ and the dots on the unit circle. The violet region in the interior of all curves is equal to the open set $\Omega$. }
		\label{fig:SpecCurvy}

	\end{figure}

\begin{rem} In the sequel, $C$ denotes a generic positive constant independent of $\cL$ and $n\in\mathbb{N}$ which might change from line to line.
\label{rem:C}
\end{rem}

By uniform continuity of the semigroup $(e^{t\cL})_{t \geq 0}$ and the upper semicontinuity of the spectrum of $M e^{t\cL}$ ~\cite[Chapter IV §3.4]{Kato}, we know that the curves $\Gamma_j$ and $\gamma$ separate parts of the spectrum of $M e^{t\cL}$  for $t>0$ small enough. We will prove this explicitly and establish a uniform bound on the resolvent of $M e^{t\cL}$ outside the interior of the curves denoted by $\Omega$. This uniform bound will be useful in proving a quantitative bound on the convergence rate \eqref{eq:QuantitativeQuantumZeno}. 
\begin{lemm}\label{lem:SpectralSep}
	For $t\in[0,\eps]$, $\cL \in \cB(X)$, and $1\le K_{\varepsilon}:=\sup_{s \in [0,\varepsilon]} \| e^{s \cL} \|$, we choose $\varepsilon>0$ such that
	\begin{align}
	\label{eq:EpsilonRange}
	0<\varepsilon < \max\Bigg\{ K_{\varepsilon}\|\cL\|\sup_{z\in B_{3/2}\setminus\Omega }\|(z-M)^{-1}\|, \|\cL\| \Bigg\}^{-1},
	\end{align}
	we have $z\in \rho(Me^{t\cL})$ for every $z\in B_{3/2}\setminus\Omega$. Moreover, there is $C>0$ such that 
	\begin{equation}
	\label{eq:UniformResolventbound}
	\sup_{(t,z)\in[0,\varepsilon]\times (B_{3/2}\setminus\Omega)}\left\|\left(z- Me^{t\cL}\right)^{-1}\right\| \le C.
	\end{equation}
\end{lemm}
\begin{proof}
For $t\in[0,\varepsilon]$ with $\varepsilon$ as in \eqref{eq:EpsilonRange} and $z\in B_{3/2}\setminus\Omega$ we have 
	\begin{align*}
	\left\|\left(z-Me^{t\cL}\right) -\Big(z-M\Big)\right\|&=\left\|M -Me^{t\cL}\right\| \le  \left\|\1 - e^{t\cL}\right\| \\&\le K_{\varepsilon} \|\cL\| t \le q \left\|(z-M)^{-1}\right\|^{-1},
	\end{align*}
	for some $0<q <1$.
	Hence, $z\in\rho(Me^{t\cL})$, and in particular
	\begin{align*}
	\left\|\left(z-Me^{t\cL}\right)^{-1}\right\| &\le(1-q)^{-1} \left\|(z-M)^{-1}\right\| \le (1-q)^{-1} \sup_{z\in\Gamma}\left\|(z-M)^{-1}\right\| \le C,
	\end{align*}
	which shows that the resolvent is uniformly bounded on $[0,\varepsilon]\times B_{3/2}\setminus\Omega$.
\end{proof}
Analogously as in \eqref{eq:SpecProjector}, we can now define for small $t>0$, as in Lemma~\ref{lem:SpectralSep}, the spectral projectors of $Me^{t\cL}$ corresponding to the part of the spectrum separated by the curve $\Gamma_j$, which is
\begin{align}
\label{eq:TimeDepProj}
P_j(t) = \frac{1}{2\pi i}\oint_{\Gamma_j} \left(z-Me^{t\cL}\right)^{-1} dz.
\end{align}
Using these projectors, we show that the main contribution to the quantum Zeno product comes from the peripheral part of the spectrum, since the contribution from the rest vanishes exponentially. This is the content of the following lemma which is a slight generalisation of Lemma $1$ of~\cite{MobWolf19}.
\begin{lemm}\label{lem:ExponVanPart}
Under the assumptions of Theorem~\ref{thm:QuantitativQuantumZeno} we have
\begin{equation*}
\left\|\left(Me^{\cL/n}\right)^n - \sum_{j=1}^J\left(P_j(n^{-1})\,Me^{\cL/n}\,P_j(n^{-1})\right)^n \right\| \le C\, \tilde\delta^{n+1} 
\end{equation*}	
for all $n\in\mathbb{N}$ such that $\eps:=n^{-1}$ satisfies \eqref{eq:EpsilonRange} and $0<\delta<\tilde\delta<1$.
\end{lemm}
\begin{proof}
Using the holomorphic functional calculus applied to the operator $Me^{\cL/n}$ we see
\begin{align*}
\sum_{j=1}^J\left(P_j(n^{-1})\,Me^{\cL/n}\,P_j(n^{-1})\right)^n = \sum_{j=1}^J \frac{1}{2\pi i}\oint_{\Gamma_j} z^n \left(z- Me^{\cL/n}\right)^{-1} dz.
\end{align*}
This implies that
\begin{align*}
&\left\|\left(Me^{\cL/n}\right)^n - \sum_{j=1}^J\left(P_j(n^{-1})\,Me^{\cL/n}\,P_j(n^{-1})\right)^n \right\| \\ &= \left\|\frac{1}{2\pi i} \oint_\gamma z^n \left(z-Me^{\cL/n}\right)^{-1}dz\right\| \le  \sup_{(t,z)\in[0,1/n]\times\gamma}\left\|\left(z-Me^{\cL/n}\right)^{-1}\right\|\,\tilde\delta^{n+1} \\&\le \sup_{(t,z)\in[0,1/n]\times (B_{3/2}\setminus\Omega)}\left\|\left(z-Me^{\cL/n}\right)^{-1}\right\|\, \tilde\delta^{n+1} \le C \,\tilde\delta^{n+1}.
\end{align*}
Here we have used the fact that $\gamma$ has distance at most $\tilde\delta$ from the origin, the uniform resolvent bound \eqref{eq:UniformResolventbound}, and the fact that the curve $\gamma$ is contained in $B_{3/2}\setminus\Omega$.
\end{proof}
In order to control the remainder term \begin{align}\label{eq:Remainder}\sum_{j=1}^J\left(P_j(n^{-1})\,Me^{\cL/n}\,P_j(n^{-1})\right)^n,\end{align} we study the derivative of $P_j(t)$ at zero which we construct in the following lemma. 

\begin{lemm}
	\label{lem:DerivativeProjector}
For each $j=1,\dots,J$ we can define the derivative in norm topology at $t=0$ of the family of projectors $P_j(t)$ (defined through~\eqref{eq:TimeDepProj}) as 
\[ P_j'=-\frac{1}{2\pi i}\oint_{\Gamma_j}\Big(z-M\Big)^{-1}M\cL\Big(z-M\Big)^{-1} dz, \]
	satisfying $\|P'_j\| \le C\|\cL\|$. Then for $t\in[0,\varepsilon]$ with $\varepsilon$ as in \eqref{eq:EpsilonRange} and some universal $C>0$, the following bound holds:
	\begin{equation}
	\label{eq:2ndDerivative}
	\left\|P_j(t) - P_j -tP_j'\right\| \le C\,t^2\|\cL\|^2.
	\end{equation}
\end{lemm}
\begin{proof}
We start by recalling that for $A,B\in\cB(X)$ and for $z\in\rho(A) \cap\rho(B)$ the difference of the resolvents is given by the {\em{second resolvent formula}}:
\begin{equation}
\label{eq:2ndResolvent}
\left(z-A\right)^{-1} - \left(z-B\right)^{-1} = \left(z-A\right)^{-1}\left(B-A\right)\left(z-B\right)^{-1}.
\end{equation}
	  Using the above we can write 
	\begin{align*}
	\frac{P_j(t)-P_j}{t} &= \frac{1}{2\pi it}\oint_{\Gamma_j}\Big(z-Me^{t\cL}\Big)^{-1} - \Big(z-M\Big)^{-1} dz \\
	&=\frac{1}{2\pi i}\oint_{\Gamma_j}\Big(z-Me^{t\cL}\Big)^{-1}\frac{M}{t}\Big(\1 -e^{t\cL}\Big)\Big(z-M\Big)^{-1} dz \\
	&\xrightarrow[t\to 0]{}\quad-\frac{1}{2\pi i}\oint_{\Gamma_j}\Big(z-M\Big)^{-1}M\cL\Big(z-M\Big)^{-1} dz.
	\end{align*}
	To justify the limit in the last line, we have used the dominated convergence theorem together with the uniform resolvent bound~\eqref{eq:UniformResolventbound}, and continuity of the operator inverse.
	To bound the norm of $P'_j$, we see that
	\begin{align*}
	\|P_j'\| &= \left\|\frac{-1}{2\pi i}\oint_{\Gamma_j}\Big(z-M\Big)^{-1}M\cL\Big(z-M\Big)^{-1} dz\right\| \\
	&\le C\sup_{z \in \Gamma_j}\Big\|\left(z-M\right)^{-1}\Big\|^2 \,\|M\|\|\cL\|  = C \|\cL\|.
	\end{align*}
 	To prove~\eqref{eq:2ndDerivative} we write using the resolvent formula~\eqref{eq:2ndResolvent}, and the bound \eqref{eq:UniformResolventbound}
 	\begin{align*}
 	&\left\|P_j(t) -P_j -tP_j'\right\| \\
 	&=\left\|\frac{1}{2\pi i} \oint_{\Gamma_j} \Big(z-Me^{t\cL}\Big)^{-1}M\Big(\1 -e^{t\cL}\Big)\Big(z-M\Big)^{-1} + \Big(z-M\Big)^{-1}M\,t\cL\Big(z-M\Big)^{-1} dz\right\| 
 	\\&\le \left\|\frac{1}{2\pi i} \oint_{\Gamma_j} \Big(z-M\Big)^{-1}M\Big(\1 + t\cL -e^{t\cL}\Big)\Big(z-M\Big)^{-1} dz\right\|
 	\\&\quad + \left\|\frac{1}{2\pi i} \oint_{\Gamma_j} \Big(z-Me^{t\cL}\Big)^{-1}M\Big(\1 -e^{t\cL}\Big)\Big(z-M\Big)^{-1}M\Big(\1 -e^{t\cL}\Big)\Big(z-M\Big)^{-1}dz\right\|
 	\\&\le C \left(\left\|\1 +t\cL - e^{t\cL}\right\|+\left\|\1 - e^{t\cL}\right\|^2\right) \le C t^2\|\cL\|^2.
  	\end{align*}
 	
\end{proof}
In order to prove convergence of the remainder term \eqref{eq:Remainder} to the quantum Zeno dynamics given by the second term on the left hand side of \eqref{eq:QuantitativeQuantumZeno}, we will employ the following strengthened version of Chernoff's $\sqrt{n}$-Lemma \cite{Chern68} which was proven in \cite{Zagr17}.
\begin{lemm}
	 \label{lem:strongChernoff} Let $Y$ be a Banach space and $K\in\cB(Y)$ a contraction. Then $\left(e^{t(K-\1)}\right)_{t\ge0}$ is a norm continuous contraction semigroup and \[\left\|\left(K^n - e^{n(K-\1)}\right)x\right\| \le 2\sqrt[3]{n}\left\|\left(K-\1\right)x\right\|\]
	 for all $x\in Y$ and $n\in\N$.
\end{lemm}
Now we have all the tools needed to conclude the proof of Theorem~\ref{thm:QuantitativQuantumZeno}. To show the quantitative bound \eqref{eq:QuantitativeQuantumZeno} we will use triangle inequality together with Lemma~\ref{lem:ExponVanPart} and, for the remainder term, the following Lemma~\ref{lem:Remainder}. Note that for proving Theorem~\ref{thm:QuantitativQuantumZeno} and in particular \eqref{eq:QuantitativeQuantumZeno} therein, we can without loss of generality assume that $n\in\N$ is large enough such that $\eps:= n^{-1}$ satisfies \eqref{eq:EpsilonRange}, since otherwise we can pick $C>0$ such that \eqref{eq:QuantitativeQuantumZeno} is trivially satisfied.
\begin{lemm} 
\label{lem:Remainder}Under the assumptions of Theorem~\ref{thm:QuantitativQuantumZeno} we have
	\begin{equation} \label{eq:Term2}
	\left\|\sum_{j=1}^J\left(P_j(n^{-1})\,Me^{\cL/n}\,P_j(n^{-1})\right)^n -  \sum_{j=1}^J e^{P_j\cL P_j}\,\lambda_j^nP_j\right\| \le C\left(\frac{\|\cL\|}{\sqrt[3]{n^2}} + \frac{\|\cL\|^2}{n}\right).
	\end{equation}
\end{lemm}
for $n\in\N$ such that $\eps:=n^{-1}$ satisfies \eqref{eq:EpsilonRange}.
\begin{proof}
	We begin the proof by defining for each $j=1,\dots,J$ the contraction
	 \begin{align}
	 K_{j,n} := \overline{\lambda_j}P_j(n^{-1})\,M e^{\cL/n}\,P_j(n^{-1}).
	 \end{align}
	We can split the left hand side of \eqref{eq:Term2} as
\begin{align}
\label{eq:Triangle}
&\nonumber\quad\,\left\|\sum_{j=1}^J\left(P_j(n^{-1})\,Me^{\cL/n}\,P_j(n^{-1})\right)^n -  \sum_{j=1}^J e^{P_j\cL P_j}\,\lambda_j^nP_j\right\|\\
&\le \sum_{j=1}^J\left\|\left(P_j(n^{-1})\, M e^{\cL/n}\,P_j(n^{-1})\right)^n - \lambda_j^n e^{n\left(K_{j,n}-P_j(n^{-1})\right)}\,P_j(n^{-1})\right\|\nonumber \\&\quad + \left\|\sum_{j=1}^J\lambda_j^ne^{n\left(K_{j,n}-P_j(n^{-1})\right)}\,P_j(n^{-1}) -  \sum_{j=1}^J e^{P_j\cL P_j}\,\lambda_j^nP_j\right\| \nonumber \\&\le \sum_{j=1}^J\Bigg(\,\left\|\left(P_j(n^{-1})\, M e^{\cL/n}\,P_j(n^{-1})\right)^n -\lambda_j^n e^{n(K_{j,n}-P_j(n^{-1}))}\,P_j(n^{-1})\right\| \nn\\&\quad + \left\|\lambda_j^n e^{n\left(K_{j,n}-P_j(n^{-1})\right)}\,P_j(n^{-1}) -  \lambda_j^n e^{P_j\cL P_j}P_j\right\|\,\Bigg),
\end{align}
and bound for all $j$ the first and second terms in the sum on the right hand side of \eqref{eq:Triangle} individually. 

For the first summand we use the refined version of Chernoff's Lemma (Lemma~\ref{lem:strongChernoff}) on Banach spaces $Y := P_j(n^{-1})X$ with induced operator norm denoted by $\|\cdot\|_Y$ (note that $P_j(n^{-1})$ corresponds to the identity on $Y$). From this we obtain
\begin{equation}
\begin{split}
\label{eq:ChernoffTerm}
\left\|\left(P_j(n^{-1})\, M e^{\cL/n}\,P_j(n^{-1})\right)^n -  \lambda_j^ne^{n\left(K_{j,n}-P_j(n^{-1})\right)}\,P_j(n^{-1})\right\|
&=\big\|K_{j,n}^n -e^{n\left(K_{j,n}-P_j(n^{-1})\right)}\big\|_Y  \\
&\le 2\sqrt[3]{n} \left\|K_{j,n} - P_j(n^{-1})\right\|_Y.
\end{split}
\end{equation}
By using the series expression of the exponential we see that
\begin{equation}
\begin{split}
\label{eq:Cheronoffexpansion}
n\left(K_{j,n} - P_j(n^{-1}) \right)&= n\overline{\lambda_j}\big(\,P_j(n^{-1})MP_j(n^{-1}) - \lambda_jP_j(n^{-1})\,\big) \\&\quad+ \overline{\lambda_j}P_j(n^{-1})M\cL P_j(n^{-1})  + \cE_{j,n},
\end{split}
\end{equation}
with $\cE_{j,n}$ being a bounded operator with $\|\cE_{j,n}\| \le C \frac{\|\cL\|^2}{n}$ containing all terms of order two or higher in the expansion. 
To bound the first term on the right-hand side of \eqref{eq:Cheronoffexpansion}, we use Lemma~\ref{lem:DerivativeProjector} and write
\begin{equation}
\begin{split}
\label{eq:FirtTermBound}
&\,\bigg\|n\bigg(P_j(n^{-1})\,M\,P_j(n^{-1})  - \lambda_jP_j(n^{-1})\bigg)\bigg\| \\
&\le \bigg\|n \left(P_jMP_j - \lambda_jP_j\right) +P_j'MP_j + P_jMP_j'  - \lambda_jP_j'\bigg\| + C\frac{\|\cL\|^2}{n}  =C\frac{\|\cL\|^2}{n}.
 \end{split}
 \end{equation}
Here we have used the fact that the first term on the second line is equal to zero. In order to see this, note that $P_jM =MP_j=\lambda_jP$, since $P_j$ is the spectral projector of $M$ corresponding to the spectral point $\lambda_j$, and the corresponding quasi-nilpotent part is zero, i.e.~$P_j$ is the projection on the corresponding eigenspace. Moreover, we used the product rule $P_j' = P_jP_j' + P_j'P_j$, which holds since all the $P_j(t)$ are projectors.

Combining \eqref{eq:FirtTermBound} with \eqref{eq:Cheronoffexpansion}, yields the following bound for \eqref{eq:ChernoffTerm}:
\begin{align*}
 \left\|\left(P_j(n^{-1}) M e^{\cL/n}P_j(n^{-1})\right)^n -  \lambda_j^ne^{n\left(K_{j,n}-P_j(n^{-1})\right)}\right\| \le  C \frac{\|\cL\|}{\sqrt[3]{n^2}}.
\end{align*}
From the above, and Lemma~\ref{lem:DerivativeProjector}, we also know that
\begin{align*}
\left\|n\left( K_{j,n} - P_j(n^{-1}) \right)- P_j\cL P_j\right\| \le  C\frac{\|\cL\|^2}{n}.
\end{align*}
Therefore, noting that both $\left(e^{tP_j\cL P_j}\right)_{t\ge 0}$ and $\left(e^{t\,n(K_{j,n}-\1)}\right)_{t\ge0}$ are contraction semigroups, we can use the bound in the proof of~\cite[Corrollary 1.11, Chapter III]{EngNag00} to infer that
\[\Big\|e^{n\left(K_{j,n}-P_j(n^{-1})\right)} -  e^{P_j\cL P_j}\Big\| \le C \frac{\|\cL\|^2}{n}.
\]
This concludes the proof.

\end{proof}

\section{Proof of Theorem \ref{thm:UnboundedQuantumZenoProjector}}
\label{sec:UnboundedQuantumZenoProjector}
In this section we consider $\cL$ to be an unbounded generator of some contraction semigroup $\left(T(t)\right)_{t\ge0}$ on a Banach space $X$, which is hence no longer uniformly continuous but only strongly continuous. Moreover, as in Section~\ref{sec:QuantitativeBoundedZeno}, we again assume the contraction $M$ to fulfill the spectral gap assumption \eqref{eq:SpectralGap0} with corresponding quasi-nilpotent operators of the spectral points on the unit circle being equal to zero. 

Under the boundedness assumption in Theorem~\ref{thm:UnboundedQuantumZenoProjector} we will prove convergence of the corresponding quantum Zeno product $\left(Me^{t\cL/n}\right)^n$. In order to do that, we will approximate the unbounded generator $\cL$ by its $k^{th}$ Yosida approximant $\cL_k$ \eqref{eq:YosidaOperator} and use the quantitative convergence rate of Theorem~\ref{thm:QuantitativQuantumZeno} for the quantum Zeno product corresponding to $\cL_k$.

\medskip

In order to prove Theorem~\ref{thm:UnboundedQuantumZenoProjector} we need some quantitative bound on the Yosida-approximation given by the following lemma. Here and henceforth we again use the convention of Remark~\ref{rem:C}.
\begin{lemm}
	\label{lem:QuantitativeBoundYosida}
	 Let $\cL$, with domain $\cD(\cL)$, be the generator of a strongly continuous contraction semigroup on some Banach space $X$, and let $B$ a bounded operator such that $B\cL$ and $\cL B$ are bounded and densely defined. Moreover, let $\cL_k$ be the $k^{th}$ Yosida approximant of $\cL$ defined in \eqref{eq:YosidaOperator}.  Then
	\begin{equation}
	\label{eq:UniformYosidaBound}
	\left\|B\cL B- B\cL_kB\right\| \le \frac{C}{k}
	\end{equation}
	and
	\begin{equation}
	\label{eq:StrongYosidaBound}
	\left\|\left(B\cL - B\cL_k \right)x\right\| \le \frac{C}{k}\|\cL x\| \text{ for all }x\in\cD(\cL).
	\end{equation}
\end{lemm}
\begin{proof}
	As $\cL$ generates a contraction semigroup, we have \cite[Theorem \@3.5, Chapter \@II]{EngNag00} 
\[\left\|\left(k-\cL\right)^{-1}\right\| \le \frac{1}{k}, \text{ for all } k \in \NN.\]
	We start by proving the uniform bound \eqref{eq:UniformYosidaBound} by first observing that
	\begin{align*}
	B\cL_kB &= Bk\cL \left(k-\cL\right)^{-1}B = B \cL\left(k-\cL\right)^{-1}\cL B + B\cL B .
	\end{align*}
	Here, we identified all bounded operators with their unique bounded extensions, i.e. $T\equiv\overline{T\vert_{\cD(T)}}$ for $T$ bounded on $\cD(T)$.  
	Using $$\left\|B \cL\left(k-\cL\right)^{-1}\cL B\right\| \le \|B\cL\| \left\|\left(k-\cL\right)^{-1}\right\| \|\cL B\| \le C/k,$$
	this shows \eqref{eq:UniformYosidaBound}. To see \eqref{eq:StrongYosidaBound}, we use the fact that
	\begin{align*}
	B\cL_k &= Bk\cL \left(k-\cL\right)^{-1} = B \cL\left(k-\cL\right)^{-1}\cL  + B\cL,
	\end{align*}
	which gives for any $x\in\cD(\cL)$
\[\left\|B \cL\left(k-\cL\right)^{-1}\cL x\right\| \le \big\|B\cL\big\|\left\|\left(k-\cL\right)^{-1}\right\|\big\|\cL x\big\| \le \frac{C}{k} \|\cL x\|.\]

\end{proof}
Using Lemma~\ref{lem:QuantitativeBoundYosida} we can prove a similar result for the corresponding semigroups.
\begin{lemm}
	\label{lem:QuantitativeBoundSemigroup}
	Let $\left(T(t)\right)_{t\ge0}$ be a strongly continuous contraction semigroup on a Banach space $X$ with generator $\cL$, and let $B$ be a bounded operator such that $B\cL$ and $\cL B$ are bounded and densely defined. Moreover, let $\cL_k$ denote the $k^{th}$ Yosida approximant of $\cL$ and $\left(T_k(t)\right)_{t\ge 0}$, with $T_k(t):= e^{t\mathcal L_k}$, be the corresponding contraction semigroup. 
	\begin{align}
	\label{eq:UniformTruncDynamicsBound}
	\Big\|BT(t)B - BT_k(t)B\Big\| \le \frac{t C}{k},
	\end{align}
	and for all $x\in\cD(\cL)$
	\begin{align}
	\label{eq:StrongTruncDynamicsBound}
	\Big\|\left(BT(t) - BT_k(t)\right)x\Big\| \le \frac{t C}{k} \|\cL x\|.
	\end{align}
\end{lemm}
\begin{proof} Since for all $k,l\in\mathbb{N}$ and $t,s\ge 0$ the operators $T_k(t)$ and $T_l(s)$ commute, it follows that 
$[T(t), T(s)]=0$. Hence, $T_k(t)$ leaves $\cD(\cL)$ invariant and commutes with $\cL$ as well. Therefore, for $x\in\cD(\cL)$ we obtain
	\begin{align*}
	\left(BT(t) - BT_k(t)\right)x &= B\int_0^t \frac{d}{d s}\,T(s)T_k(t-s)x\, d s \\
	&=\int_0^t  B\left(\cL - \cL_k\right)T(s)T_k(t-s) x\, ds.\\
	\end{align*}
	Now using~\eqref{eq:StrongYosidaBound} and the fact that $T(t)$ and $T_k(t-s)$ are contractions, we obtain \eqref{eq:StrongTruncDynamicsBound}
	\begin{align*}
	\left\|\left(BT(t) - BT_k(t)\right)x\right\| &\le  \frac{C }{k}\int_0^t \left\|\cL T(s)T_k(t-s)x\right\| ds  \le\frac{C }{k}\int_0^t \left\|\cL x\right\| ds 
	=\frac{C t}{k} \|\cL x\|.
	\end{align*} 
If we now apply this result to $x\in\cD(\cL B)$, i.e. $Bx \in \cD(\cL)$, we obtain 
	\begin{align*}
	\Big\|\left(BT(t)B - BT_k(t)B\right)x\Big\| \le \frac{t C}{k}\|\cL Bx\| \le  \frac{t C}{k} \|x\|,
	\end{align*}
	since $\cL B$ is bounded. As $\cD(\cL B)$ is dense, the uniform bound \eqref{eq:UniformTruncDynamicsBound} follows.
\end{proof}
Now we have everything needed to prove Theorem~\ref{thm:UnboundedQuantumZenoProjector}:
\begin{proof}[Proof of Theorem~\ref{thm:UnboundedQuantumZenoProjector}]
We first observe that if $M\cL$ and $\cL M$ are densely defined and bounded. This also applies to $P_j\cL, \cL P_j$, for all $j=1,\dots,J$ by using the fact that $P_j=\lambda_j^{-1} MP_j=  MP_j\lambda_j^{-1},$ since the quasi-nilpotent parts vanish. 

	Let $x\in\cD(\cL)$. Then we get by using the triangle inequality
\begin{equation}
\begin{split}
\label{eq:TriangleUnoundedcase}
&\left\|\left(\left(Me^{\cL/n}\right)^n - \sum_{j=1}^Je^{P_j\cL P_j}\lambda_j^n P_j\right)x\right\|\\ &\le \left\|\left(\left(Me^{\cL/n}\right)^n - \left(Me^{\cL_k/n}\right)^n\right)x\right\| + \left\|\left(\left(Me^{\cL_k/n}\right)^n -  \sum_{j=1}^Je^{P_j\cL_k P_j}\lambda_j^nP_j\right)x\right\| \\&\,+ 
\left\| \sum_{j=1}^J\left(e^{P_j\cL_k P_j} - e^{P_j\cL P_j}\right)\lambda_j^n P_jx\right\|.
\end{split}
\end{equation}
To bound the last term we see that for all $j\in\{1,\cdots,J\}$
\begin{align*}
&\quad\,e^{P_j\cL_k P_j} - e^{P_j\cL P_j} = \int_0^1 \frac{d}{d s}\,\left(e^{sP_j\cL_k P_j} \,e^{(1-s)P_j\cL P_j}\right)  d s \\
&= \int_0^1 e^{sP_j\cL_k P_j} \,\big(P_j\cL_k P_j- P_j\cL P_j\big)\,e^{(1-s)P_j\cL P_j} ds,
\end{align*}
and hence, using \eqref{eq:UniformYosidaBound} with $B=P_j$, we obtain $\left\| e^{P_j\cL_k P_j} - e^{P_j\cL P_j} \right\| \le C/k.$ 
Now we consider the first term on the right-hand side of \eqref{eq:TriangleUnoundedcase}. Using the fact that both $M$ and the evolution are contractions, we first notice that 
\begin{align*}
\left\|\left(\left(Me^{\cL/n}\right)^n - \left(Me^{\cL_k/n}\right)^n\right)x\right\| &\le 
\left\|\left(Me^{\cL/n} - Me^{\cL_k/n}\right)\left(Me^{\cL/n}\right)^{n-1}x\right\| \\ &\,+ \left\|\left(Me^{\cL/n}\right)^{n-1} - \left(Me^{\cL_k/n}\right)^{n-1}x\right\| \\
&\le \sum_{j=0}^{n-1} \left\|\left(Me^{\cL/n} - Me^{\cL_k/n}\right)\left(Me^{\cL/n}\right)^{j}x\right\|.
\end{align*}
For each term with $j\ge0$ in the summation above we use the uniform bound \eqref{eq:UniformTruncDynamicsBound} with $B=M$ to get
\begin{align*}
\left\|\left(Me^{\cL/n} - Me^{\cL_k/n}\right)\left(Me^{\cL/n}\right)^{j}x\right\| \le \frac{C}{nk}\|x\|.
\end{align*}
For the term with $j=0$ we use \eqref{eq:StrongTruncDynamicsBound} and get
\begin{align*}
\left\|\left(Me^{\cL/n} - Me^{\cL_k/n}\right)x\right\| \le \frac{C}{nk}\|\cL x\|.
\end{align*}
Putting these together, we obtain
\begin{align*}
\left\|\left(\left(Me^{\cL/n}\right)^n - \left(Me^{\cL_k/n}\right)^n\right)x\right\| &\le \frac{C}{nk}\left(\left\|\cL x\right\| + \sum_{j=1}^{n-1} \|x\|\right) \le \frac{C}{k} \left(\,\|x\| + \frac{\|\cL x\|}{n}\,\right).
\end{align*}
Now using Theorem~\ref{thm:QuantitativQuantumZeno} with $\cL = \cL_k$ for the middle term in \eqref{eq:TriangleUnoundedcase}, and the fact that $\|\cL_k\| \le k$, we see that for all $n\in\mathbb{N}$ 
we have
\[\left\|\left(\left(Me^{\cL/n}\right)^n - \sum_{j=1}^Je^{P_j\cL P_j}\lambda_j^nP_j\right)x\right\|\le C\Bigg(\left(\frac{1}{k} + \frac{k}{\sqrt[3]{n^2}} + \frac{k^2}{n}+\tilde\delta^{n+1}\right)\|x\| + \frac{\|\cL x\|}{nk}\Bigg),\]
for some $C>0$ and any $0<\delta<\tilde\delta<1$.
Choosing the optimal $k =\sqrt[3]{n}$, yields the bound
\[\left\|\left(\left(Me^{\cL/n}\right)^n - \sum_{j=1}^Je^{P_j\cL P_j}\lambda_j^nP_j\right)x\right\| \le C\Bigg(\left(\frac{1}{\sqrt[3]{n}} + \tilde\delta^{n+1}\right)\|x\| + \frac{\|\cL x\|}{\sqrt[3]{n^4}}\Bigg),\]
which proves \eqref{eq:UnboundedAsymptotic}. As $\left(Me^{\cL/n}\right)^n$ is a uniformly bounded sequence (in fact even a sequence of contractions) and $\cD(\cL)$ dense, this gives the strong limit stated in \eqref{eq:UnboundedLimit}.
\end{proof}

\comment{The following theorem implies that the quantum Zeno product tends to an approximate Zeno dynamics, given in terms of Yosida approximations of the full generator $\mathcal L$ under an additional constraint condition $\eqref{eq:EnergyConstraint}.$ In particular, this condition holds for any $x\in \mathcal D(\mathcal L)$ with $C_x:=1$ if $\Vert\mathcal L^{\alpha} Mx \Vert \le \Vert \mathcal L^{\alpha}x \Vert$ for all $x \in \mathcal D(\mathcal L).$}

It is desirable to find a generalisation of Theorem~\ref{thm:UnboundedQuantumZenoProjector} without the constraint that $\cL M,\,M\cL$ are bounded. A natural replacement of this uniform boundedness assumption would be a pointwise boundedness assumption depending on the corresponding element $x\in X$ on which the Zeno product is evaluated. However, in this setting it is not clear what the right candidate for the effective Zeno dynamics is since the operators $P_j\cL P_j$ are, in general, not generators of a strongly continuous contraction semigroups.

Using similar techniques as in the proof of Theorem~\ref{thm:UnboundedQuantumZenoProjector}, we  show that the Zeno product tends to an approximate Zeno dynamics, given in terms of the Yosida approximants of the full generator $\cL$. This is the statement of the following corollary, which incorporates the pointwise boundedness assumption \eqref{eq:EnergyConstraint}.
In particular, this condition holds for any $\alpha\in(0,2]$ and $x\in \mathcal D(\mathcal L^\alpha)$ with $C_x:=1$ if $\Vert\mathcal L^{\alpha} Mx \Vert \le \Vert \mathcal L^{\alpha}x \Vert$ (cf.~\cite{H06} for an overview of the theory of fractional powers of operators).
\begin{corr}
\label{theo:fakeconstraint}
	Let $\cL$, with domain $\cD(\cL)$, be a generator of a strongly continuous contraction semigroup $\left(e^{t\cL}\right)_{t\ge0}$ and $M\in\cB(X)$ be a contraction satisfying the spectral gap assumption \eqref{eq:SpectralGap0} with all corresponding quasi-nilpotent operators \eqref{eq:Nilpotent} being equal to zero. Moreover, let $\alpha \in(0,2]$ and $x\in\cD(\cL^\alpha)$ such that
	\begin{align}
	\label{eq:EnergyConstraint}
\sup_{\substack{n\in\N,\\j\in\N,\,j\le n}}\|\cL^\alpha \left(Me^{\cL/n}\right)^jx\| \le C_x,
	\end{align}
	for some finite $C_x >0$ dependent only on $x$.
	Then 
	\begin{align}
	\label{eq:UnboundedToYosida}
	\left\|\left(\left(Me^{t\cL/n}\right)^n - \sum_{j=1}^Je^{tP_j\cL_kP_j}\lambda_j^nP_j\right)x\right\| \le C_x\left(\frac{n^{1-\alpha/2}}{k^{\alpha/2}}+\frac{k}{\sqrt[3]{n^2}} + \frac{k^2}{n} + \tilde\delta^{n+1}\right),	\end{align}
	with $0<\delta<\tilde\delta<1$. In particular for $\alpha\in (4/3,\,2]$ and $x\in\cD(\cL^{\alpha})$ satisfying \eqref{eq:EnergyConstraint} we get the optimal asymptotic behaviour
	\begin{align}
	\left\|\left(\left(Me^{t\cL/n}\right)^n - \sum_{j=1}^Je^{tP_j\cL_{n^\beta}P_j}\lambda_j^nP_j\right)x\right\| \le C_x n^{-\gamma},
	\end{align}
	with $\beta =\frac{4-\alpha}{4+\alpha}$ and $\gamma = \frac{3\alpha-4}{4+\alpha}>0$ and $n$ large enough.
\end{corr}
\label{sec:channelconstraints}

\begin{proof}
	As before, we will absorb in the factor $t$ in the generator $\cL$.
	Firstly, using the triangle inequality, we split the left-hand side of \eqref{eq:UnboundedToYosida} as
	\begin{align}
	\label{eq:TriangleMiangle}
	\nonumber	\left\|\left(\left(Me^{\cL/n}\right)^n - \sum_{j=1}^Je^{P_j\cL_kP_j}\lambda_j^n P_j\right)x\right\| &\le \left\|\left(\left(Me^{\cL/n}\right)^n - \left(Me^{\cL_k/n}\right)^n\right)x\right\| \\&\quad +\left\|\left(\left(Me^{\cL_k/n}\right)^n- \sum_{j=1}^Je^{P_j\cL_kP_j}\lambda_j^nP_j\right)x\right\|.
	\end{align}
	For the first term in \eqref{eq:TriangleMiangle} we proceed as in the proof of Theorem~\ref{thm:UnboundedQuantumZenoProjector} to obtain
	\begin{align*}
	\left\|\left(\left(Me^{\cL/n}\right)^n - \left(Me^{\cL_k/n}\right)^n\right)x\right\| \le 
	\sum_{j=0}^{n-1} \left\|\left(Me^{\cL/n} - Me^{\cL_k/n}\right)\left(Me^{\cL/n}\right)^{j}x\right\|.
	\end{align*}
	Using \cite[Corollary 1.4]{GT14} this gives
	\begin{align}
	\label{eq:QuantYosBound}
	\left\|\left(\left(Me^{\cL/n}\right)^n - \left(Me^{\cL_k/n}\right)^n\right)x\right\| \le \frac{C}{(n k)^{\alpha/2} }\sum_{j=0}^{n-1} \left\|\cL^\alpha\left(Me^{\cL/n}\right)^{j}x\right\|.
	\end{align}
	Using now \eqref{eq:EnergyConstraint} we see that 
	\[\left\|\left(\left(Me^{\cL/n}\right)^n - \left(Me^{\cL_k/n}\right)^n\right)x\right\| \le \frac{C_x n^{1-\alpha/2}}{k^{\alpha/2}}.\]
	Using now for the second term in \eqref{eq:TriangleMiangle} the bound in Theorem~\ref{thm:QuantitativQuantumZeno} with bounded generator being the Yosida approximant $\cL_k$ and noting that $\|\cL_k\|\le k$, yields~\eqref{eq:UnboundedToYosida}. 
\end{proof}
\section{Proof of Theorem \ref{thm:StrongZenoGeneral}}
\label{sec:withoutspecgap}

In order to prove Theorem~\ref{thm:StrongZenoGeneral} we first introduce for all $n\in N$, $k \in [n]:=\{1,\dots,n\},$ and $\mathbf N:=(N_1,\cdots,N_{k+1})\in\N^{k+1}$ the following simplexes
\begin{equation}
    \begin{split}
    \label{eq:Ink}
I_{n,k}(\mathbf N) := \left\{i\in\N^k\Big|i_l\ge N_l\,\,\forall l\in[k], \,\sum_{l=1}^k i_l \le n-N_{k+1}\right\}
    \end{split}
\end{equation}
and analyse the asymptotic behaviour of the cardinalities of these sets, illustrated in Figure \ref{fig:simplex}, as $n \to \infty$, in the following lemma. 
\begin{figure}[t]
	 	\centering
	 		\includegraphics[width=0.55\linewidth]{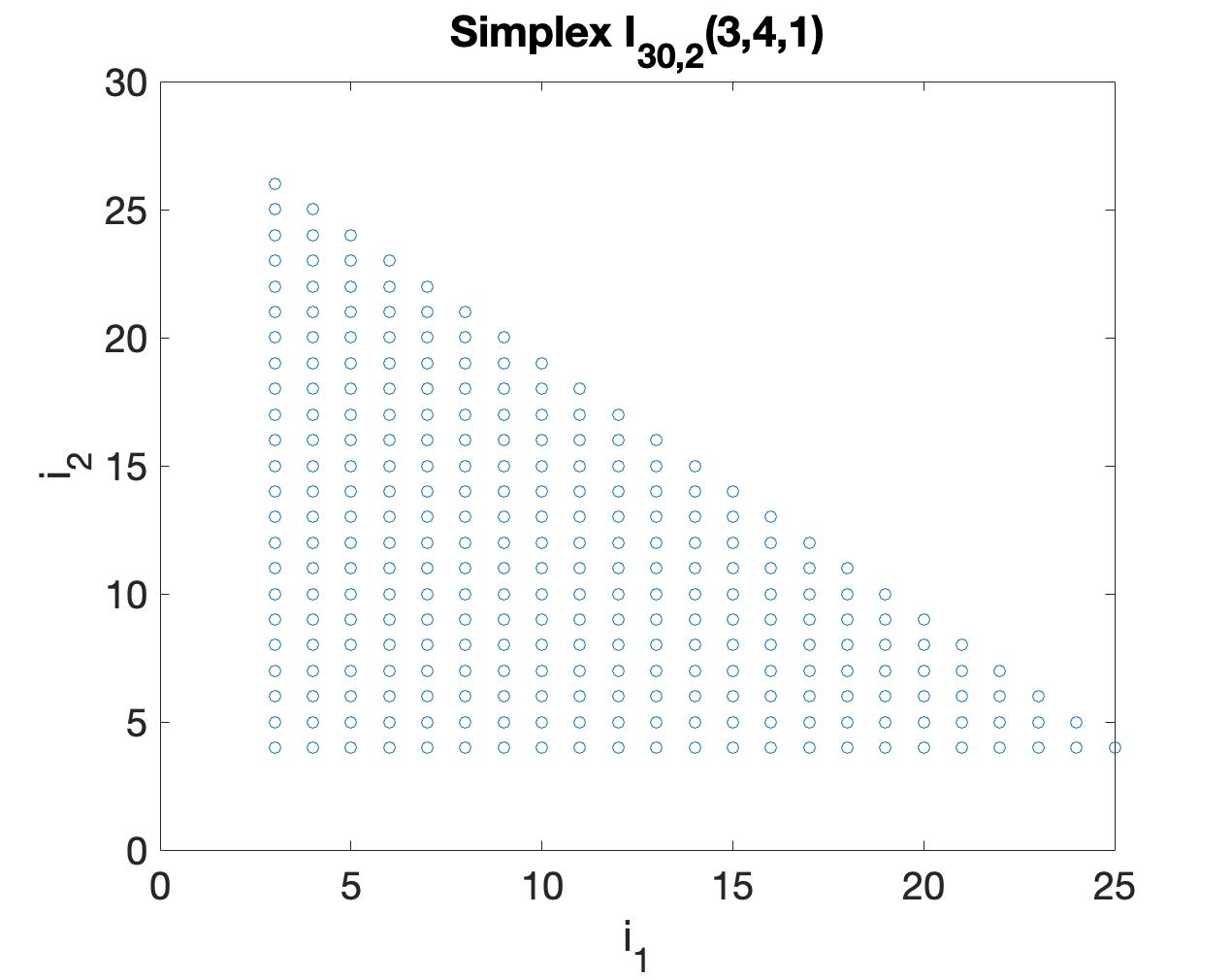}
	 	\caption{Sets $I_{n,k}(\mathbf N)$, defined in \eqref{eq:Ink} take the form of discrete simplexes.}
	 	\label{fig:simplex}
	 \end{figure}
\begin{lemm}
\label{lem:cardinality}
	For all $k\in\N$ and $\mathbf N \in\N^{k+1}$ we have
\begin{align}
\lim_{n\to\infty}\frac{|I_{n,k}(\mathbf N)|}{n^k} =\frac{1}{k!}.
\end{align}
\end{lemm}
\begin{proof}
	First we note that
	\begin{align*}
	\frac{|I_{n,k}(\mathbf N)|}{n^k} = \frac{1}{n^k} \sum_{i_1=N_1}^{n}\sum_{i_2=N_2}^{n-i_1}\cdots\sum_{i_{k-1}=N_{k-1}}^{n-\sum_{l=1}^{k-2}i_l}\sum_{i_k=1}^{n-\sum_{l=1}^{k-1}i_l-N_{k+1}}1.
	\end{align*}
	If we denote by $$\Delta^k = \left\{(t_1,\dots,t_{k})\in\mathbb{R}^{k}~\big|~\sum_{i = 1}^{k} t_i = 1 \text{ and } t_i \ge 0 \text{ for all } i\right\}$$ the $k$-simplex, then we obtain as a limiting expression, for the limit $n\to \infty$, the volume of the $k$-simplex
\begin{align*}
\lim_{n\to\infty}\frac{|I_{n,k}(\mathbf N)|}{n^k} &= \int_0^1\int^{1-t_1}_0\int_0^{1-t_1-t_2}\cdots\int_0^{1-\sum_{l=1}^{k-1} t_{l}} 1\,d t_k d t_{k-1}\cdots d t_1 \\ &=\int_{\Delta^k} 1 \ dt =\frac{1}{k!}.
\end{align*} 

\end{proof}  

For $n \in \mathbb N$ and $k \in \mathbb N$ we denote the discrete simplex by
$$\Delta_{\text{disc}}^k(n) = \left\{(i_1,\dots,i_{k})\in\mathbb{N}^{k}~\big|~\sum_{l = 1}^{k} i_l\le n \right\} = I_{n,k}(1,\cdots, 1,0).$$
We prove Theorem~\ref{thm:StrongZenoGeneral} by using a perturbation series approach. Here we split the Zeno product into a sum consisting of terms corresponding to different powers of $1/n$. In order to show convergence towards the Zeno dynamics, we need a convergence result for each of these summands which is given by the following lemma.
\begin{lemm}
	\label{lem:GeneralizedErgodic} 
		Let $\left(\cL_n\right)_{n\in\N}\subset \cB(X)$ such that
	\begin{align}
	\lim_{n\to\infty}\cL_n =\cL 
	\end{align}
	in operator norm for some $\cL\in\cB(X)$. Then for all $k\in\N$
\begin{align}
\label{eq:limitpimit}
\lim_{n\to\infty}\frac{1}{n^k}\sum_{i \in \Delta_{\text{disc}}^k(n)} M^{n+1-\sum_{l=1}^k i_{l}}\cL_nM^{i_k}\cL_n\cdots M^{i_2}\cL_nM^{i_1-1}x   = \frac{(P\cL P)^k}{k!} x
\end{align}
\end{lemm}
\begin{proof}
Let $\eps>0$. From the existence of the strong limit $\lim_{l\to\infty}M^lx=Px,$ we know that there exists a $N_1\in\N$ such that for all $i_1\ge N_1$ we have $\left\|M^{i_1-1}x-Px\right\| \le \eps.$
Using the fact that by definition $P$ is necessarily a projection, we can pick, by the same argument, for each $l\in[k]$ a $N_l(\varepsilon)\in\N$ such that for all $i_l\ge N_l(\varepsilon)$
\begin{align}
\left\|\Big(M^{i_l}\cL P\big(P\cL P\big)^{l-2}-P\big(P\cL P\big)^{l-1}\Big)x\right\| = \left\|\Big(M^{i_l}-P\Big)\cL\big(P\cL P\big)^{l-2}x\right\|\le \eps.
\end{align}
In addition, there exists a $N_{k+1}\in\N$ such that for $n$ large enough satisfying $n+1-\sum_{l=1}^k i_l \ge N_{k+1}$, i.e. $\sum_{l=1}^k i_l\le n -N_{k+1}$, we have that
\begin{align}
\left\|\Big(M^{n+1-\sum_{l=1}^k i_l}\cL\big(P\cL P\big)^{k-2}-\big(P\cL P\big)^{k-1}\Big)x\right\| \le \eps.
\end{align}
Moreover, as $\cL_n\xrightarrow[n\to\infty]{}\cL$ in operator norm, we can pick $N_\cL\in\N$ such that for all $n\ge N_\cL$ we have $\left\|\cL_n-\cL\right\| \le \eps$
and hence for each $l\in[k-1]$,
\begin{align}
\left\|\cL_n\left(P\cL P\right)^{l-1} -\cL\left(P\cL P\right)^{l-1}\right\| \le \eps \|P\cL P\|^{l-1}.
\end{align}
Now, combining the above inequalities and using $\|M\|\le1$, $\sup_{n}\|\cL_n\|\le C$ for some finite $C>0$, and the triangle inequality we get for all $(i_1,\cdots,i_k)\in I_{n,k}(\mathbf N)$, $n \ge N_\cL$, and $\mathbf N=(N_1,\cdots, N_{k+1})$ the following:
\begin{equation}
    \begin{split}
&\left\|\Big(M^{n+1-\sum_{l=1}^k i_l}\Pi_{m=k}^2(\cL_n M^{i_m}) \cL_n M^{i_1-1} -\big(P\cL P\big)^{k}\Big)x\right\| \\&\le \left\|\Big(M^{n+1-\sum_{l=1}^k i_l}\Pi_{m=k}^2(\cL_n M^{i_m}) \cL_n (M^{i_1-1} - P)\Big)x\right\| \\&\quad + \left\|\Big(M^{n+1-\sum_{l=1}^k i_l}\Pi^2_{m=k}(\cL_n M^{i_m}) \cL_n P -\big(P\cL P\big)^{k}\Big)x\right\| \\&\le C^k \eps +\left\|\Big(M^{n+1-\sum_{l=1}^k i_l}\Pi^2_{m=k}(\cL_n M^{i_m}) \cL_n P -\big(P\cL P\big)^{k}\Big)x\right\|
\end{split}
\end{equation}
Estimating the second term in the last line by using the triangle inequality and the uniform bound on operators $\cL_n$ yields
\begin{equation}
    \begin{split}
&\left\|\Big(M^{n+1-\sum_{l=1}^k i_l}\Pi^2_{m=k}(\cL_n M^{i_m}) \cL_n P -\big(P\cL P\big)^{k}\Big)x\right\|\\
&\le \left\|\Big(M^{n+1-\sum_{l=1}^k i_l}\Pi^2_{m=k}(\cL_n M^{i_m})(\cL_n - \cL )\Big)P x\right\| \\
&\quad + \left\|\Big(M^{n+1-\sum_{l=1}^k i_l} \Pi^2_{m=k}(\cL_n M^{i_m}) \cL P -\big(P\cL P\big)^{k}\Big)x\right\| \\&\le C^{k-1}\|Px\|\eps + \left\|\Big(M^{n+1-\sum_{l=1}^k i_l} \Pi^2_{m=k}(\cL_n M^{i_m}) \cL P -\big(P\cL P\big)^{k}\Big)x\right\| 
\end{split}
\end{equation}
By iterating these estimates, we obtain
\begin{equation}
\label{eq:Estimates}
    \begin{split}
&\left\|\Big(M^{n+1-\sum_{l=1}^k i_l}\Pi_{m=k}^2(\cL_n M^{i_m}) \cL_n M^{i_1-1} -\big(P\cL P\big)^{k}\Big)x\right\|\\
&\le \eps\sum_{l=1}^k\left( C^l+C^{l-1}\| (P\cL P)^{k-l}x\|\right).
    \end{split}
\end{equation}
Denoting the $n$-independent constant on the right hand side of~\eqref{eq:Estimates} as 
$$\gamma(k,P,\cL,x) = \sum_{l=1}^k\left( C^l+C^{l-1}\| (P\cL P)^{k-l}x\|\right),$$ 
we obtain, for all $n\ge N_\cL$,
\begin{align*}
&\left\|\Big(\,\frac{1}{n^k}\sum_{i \in \Delta_{\text{disc}}^k(n)} M^{n+1-\sum_{l=1}^k i_{l}}\Pi^2_{m=k}(\cL_n M^{i_m})\cL_nM^{i_1-1} - \frac{(P\cL P)^k}{k!}\Big) x\right\| \\
&=\left\|\Big(\,\frac{1}{n^k}\sum_{i\in I_{n,k}(1,\cdots,1,0)} M^{n+1-\sum_{l=1}^k i_{l}}\Pi^2_{m=k}(\cL_n M^{i_m})\cL_nM^{i_1-1} - \frac{(P\cL P)^k}{k!}\Big) x\right\| \\
&\le\left\|\,\frac{1}{n^k}\sum_{i \in I_{n,k}(\mathbf N)}\Big( M^{n+1-\sum_{l=1}^k i_{l}}\Pi_{m=k}^2 (\cL_nM^{i_m})\cL_n M^{i_1-1} - (P\cL P)^k\Big) x\right\| \\&\quad\quad+\left|\frac{|I_{n,k}(\mathbf N)|}{n^k}-\frac{1}{k!}\right|\|(P\cL P)^kx\|+\frac{|I_{n,k}(1,\cdots,1,0)|- |I_{n,k}(\mathbf N)|}{n^k}C^k\|x\| \\
&\le\frac{|I_{n,k}(\mathbf N)|}{n^k}\gamma(k,P,\cL,x) \eps +o(1).
\end{align*}
Using the fact that $\eps>0$ was arbitrary, and Lemma~\ref{lem:cardinality}, yields~\eqref{eq:limitpimit} and hence finishes the proof. 
\end{proof}
We can now give the proof of Theorem~\ref{thm:StrongZenoGeneral}:
\begin{proof}[Proof of Theorem~\ref{thm:StrongZenoGeneral}]
	As before, we can omit the time factor $t$ by absorbing it into the generator $\cL$.
	First we note that for all $n\in\N$ we can write
	\begin{align*}
	e^{\cL/n} = \1 + \frac{\cL_n}{n},
	\end{align*}
	for some sequence $\left(\cL_n\right)_{n\in\N}\subset \cB(X)$ which satisfies in operator norm
	\begin{align}
	\lim_{n\to\infty}\cL_n = \cL.
	\end{align}
	Hence for $x\in X$,
	\begin{align*}
	\left(Me^{\cL/n}\right)^nx &= \left(M\left(\1+\frac{\cL_n}{n}\right)\right)^nx = M^nx \\&\quad + \sum_{k=1}^n\frac{1}{n^k}\sum_{i \in \Delta_{\text{disc}}^k(n)} M^{n+1-\sum_{l=1}^k i_{l}}\cL_nM^{i_k}\cL_n\cdots M^{i_2}\cL_nM^{i_1-1}x.
	\end{align*}
	For each $k\in[n]$, defining
	\begin{align*}
	y_{n,k} := \frac{1}{n^k}\sum_{i \in \Delta_{\text{disc}}^k(n)} M^{n+1-\sum_{l=1}^k i_{l}}\cL_nM^{i_k}\cL_n\cdots M^{i_2}\cL_nM^{i_1-1}x,
	\end{align*}
	we see by Lemma~\ref{lem:GeneralizedErgodic} that
	\begin{align*}
	\lim_{n\to\infty} y_{n,k} = \frac{\left(P\cL P\right)^k}{k!}x.
	\end{align*}
	Moreover, by using the facts that $\|M\|\le1$ and $\sup_{n\in\N}\|\cL_n\|\le C$ for some finite $C>0$, together with the argument in the proof of Lemma~\ref{lem:cardinality} we see that
	\begin{align*}
	\|y_{n,k}\| \le \frac{|I_{n,k}(1,\cdots,1,0)|}{n^k} C^k \|x\| \le 2\frac{C^k}{k!}.
	\end{align*}
	Hence, by the dominated convergence theorem we obtain
	\begin{align*}
	\lim_{n\to\infty}\left(Me^{\cL/n}\right)^nx &=\lim_{n\to\infty}\Big( M^nx + \sum_{k=1}^n y_{n,k} \Big)=Px + \sum_{k=1}^\infty \frac{\left(P\cL P\right)^k}{k!}x = e^{P\cL P}Px, 
	\end{align*}
	which finishes the proof.

\end{proof}
\section{Methods to prove strong power-convergence in trace norm}
\label{sec:strongpowconv}
\label{sec:HSnorm}
In this section, we study two different conditions on quantum channels $M$ which obey the power-convergence assumption \eqref{eq:StrongPowerConv} in Theorem \ref{thm:StrongZenoGeneral}. We complete this theoretical study by proving that this strong power-convergence property is satisfied by a variety of physically relevant examples of quantum channels $M$. 

For a quantum Markov semigroup (QMS) $(T(t))_{t \ge 0}$ acting on bounded linear operators (Heisenberg picture), we want to study when the channel $M:=T_*(t_0)$ of the associated predual semigroup $(T_*(t))_{t \ge 0}$, acting on density operators (Schrödinger picture), evaluated at a fixed time $t_0>0$ is strongly power-convergent.

Since the large $n$ limit of $M^n(\rho)$ for a density operator $\rho$ is equivalent to the study of the large $t$ limit of $T_*(t)(\rho),$ we study two methods that imply strong pointwise convergence of the predual semigroup to an invariant state.  

The first approach in Section \ref{sec:HS} embeds the QMS into the space of Hilbert-Schmidt operators and uses Hilbert space techniques to analyze the ergodic properties of the semigroup. We then show in Lemma \ref{lem:getthatstrongconv} how this study on Hilbert-Schmidt operators can be extended to trace-class operators.

The second approach in Section \ref{sec:SEQ} relies on ergodic methods for von Neumann algebras and establishes that, under conditions on the commutant of Lindblad operators and the Hamiltonian, the predual semigroup is strongly ergodic.

\subsection{Embedding into Hilbert-Schmidt operators}
\label{sec:HS}
The first method we will discuss relies on the approach developed in \cite{CF00}. Here we use an embedding into the Hilbert space of Hilbert-Schmidt operators and then infer from convergence results in that Hilbert space, convergence results for trace-class operators.

We say that a state $\rho$ is faithful if $\operatorname{tr}(\rho x)=0 $ for $x\ge0$ implies $x=0.$ Let $\rho$ be a faithful state then we can define the embedding 
\[ i_{\rho}: \cB(\cH) \rightarrow \operatorname{HS}(\cH), \quad i_{\rho}(x):=\rho^{\frac{1}{4}} x \rho^{\frac{1}{4}}.\]
For any operator $T:\cB(\cH) \rightarrow \cB(\cH)$ satisfying the Schwartz property, i.e.
\begin{align}
T(x^*)T(x) \le T(x^* x),\quad \forall x\in\cB(\cH),
\end{align}
and 
\begin{align}
\tr\left(\rho T(x)\right) \le \tr\left(\rho x\right), \quad\forall x\in\cB(\cH),
\end{align}
we can define an operator $T^{\operatorname{HS}}$ on the dense subspace $i_{\rho}(\cB(\cH)) \subset \operatorname{HS}(\cH)$ by
\begin{align}
\label{eq:HSembedding}
 T^{\operatorname{HS}} \circ i_{\rho} = i_{\rho} \circ T
\end{align}
 and then uniquely extend it to a contraction $ T^{\operatorname{HS}}: \operatorname{HS}(\cH) \rightarrow \operatorname{HS}(\cH)$ (cf. \cite[Proposition 2.1, Proposition 2.2]{CF00}). Consequently, for $\left(T(t)\right)_{t\ge 0}$ being a semigroup of completely positive operators on $\cB(\cH)$ with invariant state $\rho$, i.e.~$\tr\left(\rho T(t)(x)\right)=\tr\left(\rho x\right)$ for all $x\in \cB(\cH)$, we can define the contraction semigroup $\big( T^{\operatorname{HS}}(t)\big)_{t\ge 0}$ on $\operatorname{HS}(\cH)$ by \eqref{eq:HSembedding}. Moreover, $\big( T^{\operatorname{HS}}(t)\big)_{t\ge 0}$ is strongly continuous if $\left(T(t)\right)_{t\ge 0}$ is weak$^*$ continuous \cite[Theorem 2.3]{CF00}.

Note also that this approach is equivalent to extending the semigroup onto the weighted $L^2(\rho)$ space \cite{Ko84,OZ99,KT13} given by the completion of $\cB(\cH)$ in the norm $\|i_\rho(\cdot)\|_2$ .
 
 The following lemma shows that we can conclude convergence to the invariant subspace of the semigroup $\left(T(t)\right)_{t\ge 0}$ on the bounded operators from the corresponding semigroup $\left(T^{\operatorname{HS}}(t)\right)_{t\ge 0}$ on $\operatorname{HS}(\cH).$

\begin{lemm}
\label{lem:getthatstrongconv}
 Let $\left(T(t)\right)_{t\ge0}$ be a semigroup of completely positive contractions on $\cB(\cH)$ with faithful invariant state $\rho$. Denote by  $\big( T^{\operatorname{HS}}(t)\big)_{t\ge 0}$ the corresponding contraction semigroup on $\operatorname{HS}(\cH)$ uniquely defined by $$ T^{\operatorname{HS}}(t)\circ i_\rho = i_\rho\circ T(t).$$ If for all $x\in\operatorname{HS}(\cH)$ $$w-\lim_{t\to\infty} T^{\operatorname{HS}}(t)(x) = \tr\left(\sqrt{\rho}\,\sigma\right)\sqrt{\rho}$$ in $\operatorname{HS}(\cH)$ then for all $x\in\cB(\cH)$
 \begin{align}
 \label{eq:wstarconv}
 w^*-\lim_{t\to\infty}T(t)(x)= \tr\left(\rho x\right)\1,
 \end{align}
 in $\cB(\cH).$ Consequently if moreover for each $t\ge0$ the operator $T(t)$ is unital and weak$^*$ continuous, we have that the predual semigroup $\left(T_*(t)\right)_{t\ge0}$ on $\cT(\cH)$ converges strongly, i.e.
 \begin{align}
 \label{eq:strongpredual}
\lim_{t\to\infty}T_*(t)(x) = \tr(x)\rho
 \end{align}
 for all $x\in\cT(\cH).$
\end{lemm}
\begin{proof}
Define the bounded projection $P^{\operatorname{HS}}$ on $\operatorname{HS}(\cH)$ by $ P^{\operatorname{HS}}(\sigma) = \tr(\sqrt{\rho}\sigma)\sqrt{\rho}$ for any $\sigma\in\operatorname{HS}(\cH)$ and moreover the bounded projection $P$ on $\cB(\cH)$ by $P(x) = \tr(\rho x) \1$ for any $x\in\cB(\cH)$. Note that we have
\begin{align}
i_\rho\circ P = P^{\operatorname{HS}} \circ i_\rho.
\end{align}
Let $y$ be a finite rank operator on $\cH$ such that the vectors $\{e_n\}_{n=1}^m,\{f_n\}_{n=1}^m\subset\cH$ in its singular value decomposition
\begin{align}
\label{eq:singval}
y = \sum_{n=1}^m \mu_n \ket{e_n}\bra{f_n}
\end{align}
are in the dense domain of $\rho^{-1/4}$ denoted by $\cD(\rho^{-1/4})$, i.e.
\begin{align}
\label{eq:rhovierteldomain}
\{e_n\}_{n=1}^m,\{f_n\}_{n=1}^m\subset\cD(\rho^{-1/4}).
\end{align}
This shows that $\tilde y = \rho^{-1/4}y\rho^{-1/4}$ is well-defined and has finite rank. Hence, we see that for each $x\in\cB(\cH)$
\begin{align*}
\tr\Big(y\, T(t)(x)\Big) &= \tr\Big(\tilde y\, i_\rho\big( T(t)(x)\big)\Big) = \tr\Big( \tilde y \, T^{\operatorname{HS}}(t)\big(i_\rho(x)\big)\Big) \\  &\xrightarrow[t\to\infty]{}\tr\Big( \tilde y \, P^{\operatorname{HS}}\big(i_\rho(x)\big)\Big)  = \tr\Big(\tilde y\, i_\rho\big(P(x)\big)\Big) = \tr\Big(y\,P(x)\Big).
\end{align*}
 as $ T^{\operatorname{HS}}(t)(i_\rho(x))$ is weakly convergent to $ P^{\operatorname{HS}}(i_\rho(x))$. As the set of all finite rank operators $y$ satisfying \eqref{eq:singval} and \eqref{eq:rhovierteldomain} is dense in the trace-class operators $\cT(\cH)$, $\left(T(t)\right)_{t\ge0}$ is a semigroup of contractions and $P$ a bounded operator on $\cB(\cH)$, this already shows \eqref{eq:wstarconv}.
 
 Noting now that the operator on $\cT(\cH)$ defined by $P_*(x) := \tr(x)\rho$ is the predual of $P$, \eqref{eq:wstarconv} directly gives the weak convergence of the predual semigroup,
 i.e.
 \begin{align*}
 w-\lim_{t\to\infty}T_*(t)(x) = P_*(x) = \tr(x)\rho
 \end{align*}
 for all $x\in\cT(H).$ Moreover, let now $x\in\cT(\cH)$ be positive semidefinite, i.e. $x\ge0$. As for each $t\ge0$ the operators $T_*(t)$ and $P_*$ are completely positive and trace-preserving this gives
 \begin{align*}
\|T_*(t)(x)\|_1 = \tr\big(T_*(t)(x)\big) = \tr\big(x\big) = \tr\big(P_* (x)\big) = \|P_*(x)\|_1.
 \end{align*}
 Hence, by \cite{A81} we can conclude that for all positive semidefinite $x\in\cT(\cH)$
 \begin{align*}
     \lim_{t\to\infty} T_*(t)(x) = P_*(x)
     \end{align*}
 in trace-norm. As every trace-class operator can be written as a linear combination of four positive semidefinite trace-class operators this shows \eqref{eq:strongpredual}.
\end{proof}
We will use the construction above to show strong power-convergence for quantum channels $M$. The main idea here is to start with the dual channel in the Heisenberg picture and transform it into a contraction on the space of Hilbert-Schmidt operators using the embedding $i_\rho$. On the latter space one can then show uniform convergence towards the invariant subspace if a certain spectral gap condition is satisfied.

\comment{ In particular we will use the following proposition to show convergence (cf.~\cite[Theorem 2.3]{L89}, \cite[Proposition 2.5]{CF00} and \cite[equation 9]{CFGQ08} for proof and discussion).
\begin{prop}
Let $\cL$, with domain $\cD(\cL)$, be the generator of a strongly continuous contraction semigroup $\left(T(t)\right)_{t\ge 0}$ on a Hilbert space $\mathcal{K}$. We define the spectral gap of $\cL$ by
\begin{align}
\operatorname{gap}\left(\cL\right) = \inf\left\{ -\Re\langle x,\cLx\rangle \Big|\, x\in\cD(\cL),\,\|x\|=1,\,x\in\left(\operatorname{ker}(\cL)\right)^\perp\right\}. 
\end{align}
Then $\operatorname{gap}(\cL)$ is the biggest number $\lambda$ for which for every element $x\in\mathcal{K}$
\begin{align}
 \left\|T(t)x - Px\right\| \le e^{-\lambda t} \left\|x - Px\right\|, 
\end{align}
where $P$ is the orthogonal projection onto $\operatorname{ker}(\cL).$
\end{prop}
}

\begin{example}[Quantum Ornstein-Uhlenbeck semigroup,\cite{CFL00}] 
\label{ex:qOU}
The quantum Ornstein-Uhlenbeck semigroup (qOU) in the Heisenberg picture, i.e. on $\cB(\cH)$, is generated by 
\begin{align}
\cL x = -\frac{\mu^2}{2}\left(a^*a x +x a^*a - 2a^*xa\right) - \frac{\nu^2}{2}\left(aa^*x + xaa^* - 2axa^*\right)
\end{align}
and will be denoted by $\left(T(t)\right)_{t\ge0}.$
Here $0<\lambda<\mu$ and $a^*$ and $a$ are the creation and annihilation operators satisfying the canonical commutation relations $[a,a^*]= \1$. The qOU semigroup arises in quantum optics models of masers and lasers, and in weak-coupling models of open quantum systems. The faithful invariant state of the qOU semigroup is given by 
\[ \rho = (1-\nu) \sum_{n \ge 0} \nu^n \vert n \rangle \langle n \vert \]
where $\nu= \frac{\lambda^2}{\mu^2}$ and $\left\{\ket{n}\right\}_{n\ge 0}$ denotes the eigenbasis of the number operator $N =a^*a.$

Let $\left(T^{\operatorname{HS}}(t)\right)_{t\ge0}$ be the corresponding contraction semigroup on the Hilbert-Schmidt operators defined by
\begin{align*}
T^{\operatorname{HS}}(t)\circ i_\rho = i_\rho\circ T(t).
\end{align*}
The generator of $T^{\operatorname{HS}}(t)$ denoted by $\cL^{\operatorname{HS}}$ is self-adjoint, has compact resolvent, and can be explicitly diagonalized as
\[ \cL^{\operatorname{HS}} =- \left(\frac{\mu^2-\lambda^2}{2} \right) \sum_{n\ge 0} n P_{E_n} \]
where $P_{E_n}$ is the orthogonal projection onto 
\[ E_n:=\operatorname{span}\{ \rho^{1/4} p_n(Q_z) \rho^{1/4}: \vert z \vert=1\}\]
with $Q_z= 2^{-1/2}\left( \bar z a + za^*\right)$, polynomials $p_n$ given by 
\[ p_n(t) = \sum_{2r\le n} \left(-\frac{\mu^2+\lambda^2}{4(\mu^2-\lambda^2)} \right)^r \frac{n!}{r!(n-2r)!}t^{n-2r},\]
and eigenvalues $\left(-n\left(\frac{\mu^2-\lambda^2}{2}\right)\right)_{n \ge 0}.$ 

By the spectral mapping theorem we have  $$\Spec\left(T^{\operatorname{HS}}(t)\right)=\Spec\left(e^{  \cL^{\operatorname{HS}}t} \right)= \exp\left(-nt\left(\frac{\mu^2-\lambda^2}{2}\right) \right)_{n \ge 0}.$$ Using Corollary~\ref{cor:uniformpowerconv} and the fact that $T^{HS}(t)$ is self-adjoint and hence the corresponding quasi-nilpotent operator \eqref{eq:Nilpotent} at the isolated spectral point 1 is zero, this implies that 
\begin{align*}
\lim_{t\to\infty} T^{\operatorname{HS}}(t) = P_{E_0} 
\end{align*}
uniformly in $\operatorname{HS}(\cH)$, with $P_{E_0}$ being explicitly given by $P_{E_0}(x)= \tr\left(\sqrt{\rho}x\right)\sqrt{\rho}$. Hence, using Lemma~\ref{lem:getthatstrongconv} we see that for the qOU semigroup in the Heisenberg picture we get 
\begin{align*}
w^*-\lim_{t\to\infty}T(t)(x) = \tr\left(\rho x\right) \1 
\end{align*}
for all $x\in\cB(\cH)$, and strong convergence of the corresponding qOU semigroup in the Schrödinger picture given by the predual $T_*(t)$
\begin{align*}
\lim_{t\to\infty}T_*(t)(x) = \tr\left( x\right)\rho.
\end{align*}
Hence, for any fixed $t_0>0$ the quantum channel $M:=T_*(t_0)\in\cB\left(\cT(\cH)\right)$ is strongly power-convergent and satisfies the condition \eqref{eq:StrongPowerConv} in Theorem~\ref{thm:StrongZenoGeneral}.
\end{example}
\subsection{Strongly ergodic quantum Markov semigroups}
\label{sec:SEQ}
In this section we discuss the ergodic approach of \cite{FV82,DFR10} to identify strongly ergodic predual QMS.

Consider the minimal Quantum Markov semigroup $( T(t))_{t \ge 0}$ (see \cite{DFR10} for definitions) acting on the space of bounded linear operators $\mathcal B(\mathcal H)$, 
whose Lindblad operators $(L_l)$ are closed and have domains $D(L_l) \subset D(G)$, where $G$ is a generator of some $C_0$-semigroup such that 
\begin{itemize}
    \item For all $u,v \in D(G)$
    \[\langle Gv, u \rangle +\langle v, Gu \rangle + \sum_{l \ge 1 } \langle L_l v,L_l u \rangle=0.\]
    \item There exists a dense linear subspace $D$ of $\cH$ such that $D \subset D(G) \cap D(G^*) \cap D(L_l) \cap D(L_l^*)$ such that
    \begin{itemize}
        \item The operator $H=(G-G^*)/2$ is essentially self-adjoint and the unitary group $(e^{itH})_{t \in \mathbb R}$ satisfies $e^{itH}(D) \subset D(G)$ for all times $t \in \mathbb{R}.$
        \item The operator $G_0$ defined on $u \in D$ by $G_0=(G+G^*)/2$ is essentially self-adjoint and $D(G) \subset D(G_0) \subset D(L_l)$ for all $l \ge 1.$
    \end{itemize}
\end{itemize}

We define the fixed-point algebra of bounded linear operators left invariant by the QMS \[\mathcal F(T):=\{ X \in \cB(\cH): T(t)(X)=X \text{ for all } t \ge 0 \}\] and the decoherence-free subalgebra 
\begin{equation}
    \begin{split}
        \mathcal N(T):=\{&X \in \cB(\cH):T(t)(X^*X)=T(t)(X^*)T(t)(X) \text{ \& }\\ &T(t)(XX^*)=T(t)(X)T(t)(X^*)\}.
    \end{split}
\end{equation}

Let $\rho$ be a faithful normal invariant state of the predual semigroup $(T_*(t))_{t\ge 0}$. If the fixed-point algebra and decoherence-free algebra coincide, i.e. $\mathcal F(T)=\mathcal N(T),$ then the predual semigroup satisfies by a theorem due Frigerio and Verri, \cite{FV82}, that
\begin{equation}\lim_{t \rightarrow \infty} T_*(t)(\sigma)=\rho \text{ for all states } \sigma.
\end{equation}

A useful commutator condition to verify $\mathcal F(T)=\mathcal N(T)$ for practical examples of Quantum Markov semigroups has been identified in \cite[Theo $3.3$]{DFR10}.

In fact the above criterion can be applied to identify the following physically relevant strongly convergent quantum dynamical semigroups (QDS) \cite[4.5,4.6]{FR}, by which we mean the predual semigroup (Schr\"odinger picture) associated to the minimal QMS.

\begin{example}[Jaynes-Cummings model]
\label{ex:Jaynes-Cumming}
The quantum Markov semigroup for the Jaynes-Cummings model is defined using Lindblad operators $L_1=\mu a, L_2= \lambda a^*, L_3=R \cos(\phi \sqrt{aa^*})$, and $L_4=Ra^*\frac{\sin(\phi \sqrt{a^*a})}{\sqrt{a^*a}}$ with parameters $\varphi,R \ge 0$ and $\lambda< \mu.$

The semigroup then has a stationary state given by 
\[\rho_{\infty}:=\sum_{n=0}^{\infty} \pi_n \vert e_n \rangle \langle e_n \vert\]
where 
\[ \pi_n=c \prod_{k=1}^n \frac{\lambda^2k + R^2 \sin^2(\phi \sqrt{k})}{\mu^2 k},\] 
with normalization constant $c>0.$
\end{example}

\begin{example}[Emission-Absorption process]\label{ex:EmAbs}
The emission-absorption model is defined using Lindblad operators $L_1 = \nu a^*a$, $L_2=\mu a$ and Hamiltonian $H=\xi(a+a^*)$ where $\mu,\nu>0$ and $\xi \in \RR$. It follows from \cite[Corrollary $6.3$]{FR} that this QDS is strongly convergent to a unique invariant state.
\end{example}

\begin{example}[Two-photon absorption and emission process,\cite{FQ05,CFGQ08}] 
\label{ex:photonabs}
The two photon absorption process is the simultaneous absorption of two photons by molecules or atoms. In the Heisenberg picture, the coupling of the one-mode electromagnetic field with a bosonic gaussian positive temperature reservoir of two-photon absorbing atoms is described by the following generator in terms of operators $b:=a^2$, where $a$ is the usual annihilation operator
\begin{equation}
\mathcal Lx = i \kappa[b^{*}b,x]- \frac{\mu^2}{2}(b^*bx  + x b^*b- 2 b^*x b )-\frac{\lambda^2}{2}(bb^*x  + x b^*b - 2 bx b^*)
\end{equation}
with $\mu^2 = e^{\beta \omega}/(e^{\beta \omega}-1)$ and $\lambda^2 = 1/(e^{\beta \omega}-1)$ are the absorption and emission rates with characteristic frequency $\omega$ and inverse temperature $\beta.$
Writing $\left(T(t)\right)_{t\ge0}$ for the semigroup on $\cB(\cH)$ and $\left(T_*(t)\right)_{t\ge0}$ for the corresponding predual semigroup in the Schrödinger picture, it was shown in \cite[Proposition 7.1]{FQ05} that
\begin{align}
\lim_{t\to\infty}T_*(t)(x) = \tr\left(\Pi_e x\right)\rho_e + \tr\left(\Pi_{\rho_o} x\right)\rho_o.
\end{align}
Here $\rho_e = (1-\nu^2) \sum_{k \ge 0} \nu^{2k} \vert 2k \rangle \langle 2k \vert$ and $\rho_0 = (1-\nu^2) \sum_{k \ge 0} \nu^{2k} \vert 2k+1 \rangle \langle 2k+1 \vert$ with $\nu=\lambda/\mu$ form a basis of the invariant subspace of $T_*(t)$ and $\Pi_e$ and $\Pi_o$ denote the projections onto their support. Hence, for any fixed time $t_0>0$ the quantum channel $M := T_*(t_0)\in\cB(\cT(\cH))$ is strongly power-convergent and hence satisfies the assumption \eqref{eq:StrongPowerConv} in Theorem~\ref{thm:StrongZenoGeneral}.

\comment{
For $\alpha \in (0,1)$ we then define
\[ \rho= \alpha \rho_e+ (1-\alpha)\rho_0\]
with $\rho_e = (1-\nu^2) \sum_{k \ge 0} \nu^{2k} \vert 2k \rangle \langle 2k \vert$ and $\rho_0 = (1-\nu^2) \sum_{k \ge 0} \nu^{2k} \vert 2k+1 \rangle \langle 2k+1 \vert$ where $\nu=\lambda/\mu.$

The spectral gap of the generator on Hilbert-Schmidt operators, $ L$, is for any set of parameters $0 < \lambda \le \mu$ given in terms of the emission rate only and is equal to $\lambda^2>0.$ In particular, for $x \in \operatorname{HS}(\mathcal H)$ let $Px= \langle \sqrt{\rho_0}, x \rangle \sqrt{\rho_0}+ \langle \sqrt{\rho_e}, x \rangle \sqrt{\rho_e}$ be the projection onto $\ker(L)$ then
\[ \Vert (e^{Lt}-P)x \Vert_{2} \le  e^{-\lambda^2 t} \Vert x-Px \Vert_2. \]
}
\end{example}

Further examples of quantum dynamical semigroups that converge strongly to an invariant state include the quadratic open quantum harmonic oscillator \cite{DFY20}.

\comment{
\begin{example}[Quantum Brownian motion, \cite{CFL00}]
\label{ex:qBM}
In the above Example \ref{ex:qOU}, if we choose $\lambda=\mu$, the invariant state $\rho$ ceases to exist. However, one can still assign a generator to this configuration on Hilbert-Schmidt operators by considering the Dirichlet form associated to $-L$, defined in Example \ref{ex:qOU}, given by 
\begin{equation}
    \mathcal E(\xi):=\frac{1}{2} \left( \Vert \mu a\xi - \lambda \xi a\Vert^2 + \Vert \mu a \xi^* - \lambda a \xi^* \Vert^2 \right).
\end{equation}
For $\lambda=\mu$ this Dirichlet form reduces to 
\[ \mathcal E(\xi):= \lambda \Vert a\xi-\xi a \Vert^2.\]
Thus, there is a positive self-adjoint operator $L_{\operatorname{BM}}$ associated to this Dirichlet form on Hilbert-Schmidt operator. The operator $-L_{\operatorname{BM}}$ generates a contraction semigroup associated to \emph{quantum Brownian motion}. The spectrum of $L_{\operatorname{BM}}$ can be explicitly computed \cite[Theo $8.1$]{CFL00} and is \[ \Spec(L_{\operatorname{BM}})=[0,\infty). \]
\end{example}
}

\section{Open problems}
\label{sec:open}
In the following, we list various questions that have only been partly addressed in this article or would require tools beyond the scope of this article:

\begin{enumerate}
\item Throughout the entire article, we only consider quantum channels $M$ with finite point spectrum on the unit circle. It would be desirable to develop tools which would also allow the study of the quantum Zeno effect for quantum channels with continuous point spectrum on the unit circle. 

\item In finite dimensions it is well-known that quantum channels can only have discrete spectrum on the unit circle with vanishing nilpotent parts. It would be interesting to see whether this property also holds for quantum channels acting on infinite-dimensional spaces or whether there exist examples which violate this property.

\item While our results are the first to provide quantitative convergence rates for infinite-dimensional quantum channels, it would be natural to investigate whether these convergence rates are optimal.

\item It would also be interesting to see whether Theorem \ref{thm:StrongZenoGeneral} could be extended to the case of unbounded generators.

\item We saw that the generator studied in Example \ref{ex:qOU} \comment{and \ref{ex:qBM}} is a self-adjoint operator on the Hilbert space of Hilbert-Schmidt operators. In particular, for such generators we have the following convergence result \cite[Corr. $1$]{MaShvi03}:
\begin{prop}\
\label{prop:conv}
Let $(-A)$ be the generator of a holomorphic strongly continuous semigroup on a Hilbert space $\cH$, where 
\[ \Vert e^{-z A} \Vert \le 1 \text{ for all } z \in \{\xi \in \mathbb C \setminus \{0\}; \vert \operatorname{arg}(\xi)\vert < \tau\}, \tau \in (0,\pi/2].\]
Let $P$ be an orthogonal projection, then there exists a continuous degenerate semigroup $(S(t))_{t \geq 0}$, i.e.~$S(0)$ is a bounded projection such that $S(t)$ is strongly continuous on $S(0)\cH,$
\[ S(t)x = \lim_{n \rightarrow \infty} (e^{-tA/n}P)^n x \text{ for all } x \in \cH.\]
\end{prop}
This result overcomes the issue of explicitly identifying the generator of the quantum Zeno dynamic. It would be interesting to see if a similar result also holds on spaces of trace-class operators for quantum dynamical semigroups.
\end{enumerate}


\begin{thebibliography}{9}
	
\bibitem[A81]{A81} Arazy, J. \emph{More on Convergence in Unitary Matrix Spaces},	Proceedings of the American Mathematical Society
Vol. 83, No. 1, pp. 44-48, 1981.
	
\bibitem[BFNPY18]{BuFaNaPaYu18} Burgarth, D., Facchi, P., Nakazato, H., Pascazio, S. and Yuasa, K., \emph{Quantum Zeno dynamics from general quantum operations}, \arXiv{1809.09570}, 2018. 


\bibitem[Ch68]{Chern68} Chernoff, P. R., \emph{Note on product formulas for operator semigroups,} Journal of Functional Analysis, Vol. 2, pp. 238–242, 1968.

\bibitem[CF00]{CF00} Carbone, R. and Fagnola, F., \emph{Exponential $L^2$-convergence of quantum
Markov semigroups on $\mathcal B(h)$,} Math. Notes 68, no. 3-4, 452-463, 2000.

\bibitem[CFGQ08]{CFGQ08} Carbone, R., Fagnola, F. , Garcia, J. C., and Quezada, R., \emph{Spectral properties of the two-photon
absorption and emission process.} J. Math. Phys. 49, 032106, 2008.

\bibitem[DTG16]{DTG16} De Palma, G., Trevisan, D., and Giovannetti, V. \emph{Passive states optimize the output of bosonic gaussian quantum channels}. IEEE Transactions on Information Theory, Volume:  62, Issue:  5, 2016.

\bibitem[C15]{Chang} Chang, M. \emph{Quantum Stochastics}.  Cambridge University Press, 2015. 

\bibitem[CFL00]{CFL00} Cipriani, F., Fagnola, F. \text{ and } Lindsay, J.M., \emph{Spectral Analysis and Feller Property for Quantum Ornstein–Uhlenbeck Semigroups.} Communications in Mathematical Physics volume 210, pages85–105, 2000.

\bibitem[DFR10]{DFR10} Dhari, A. , Fagnola, F., and Rebolledo, R. \emph{The decoherence-free subalgebra of a quantum Markov semigroup with unbounded generator},
Infinite Dimensional Analysis, Quantum Probability
and Related Topics
Vol. 13, No. 3 (2010) 413–433

\bibitem[DFY20]{DFY20} Dhari, A. , Fagnola, F., and Yoo, H. \emph{Quadratic open quantum harmonic oscillator},
Letters in Mathematical Physics (2020) 110:1759-1782

\bibitem[EARV04]{EARV04} Erez, N., Aharonov, Y., Reznik, B., and Vaidman, L., \emph{Correcting quantum errors with the Zeno effect}, Phys. Rev. A 69, 062315, 2004.

\bibitem[EN00]{EngNag00} Engel, K-J. and Nagel, R. \emph{One-Parameter Semigroups for Linear Evolution Equations.} Springer. Graduate Texts in Mathematics, 2000.

\bibitem[E205]{Ex97} Pavel Exner. Unstable system dynamics: do we understand it fully? Rep. Math. Phys.,
59(3):351–363, 2007.

\bibitem[EI05]{ExIch05} Exner, P., Ichinos, \emph{A Product Formula Related to Quantum Zeno Dynamics},
Ann. Henri Poincar\'e 6, pp. 195 – 215, 2005.

\bibitem[FJP04]{FJP04} Franson, J. D., Jacobs, B. C. and Pittman, T. B., \emph{Quantum computing using single photons and the Zeno effect}, Phys. Rev. A , vol. 70, 062302, 2004.

\bibitem[FP08]{FP} Facchi, P. and Pascazio, S. (2008). \emph{Quantum Zeno dynamics: mathematical and physical aspects}. Phys. A: Math. Theor. 41 493001, Topical Review

\bibitem[FR06]{FR} Fagnola F., Rebolledo R. (2006) \emph{Notes on the Qualitative Behaviour of Quantum Markov Semigroups.} In: Attal S., Joye A., Pillet CA. (eds) Open Quantum Systems III. Lecture Notes in Mathematics, vol 1882. Springer, Berlin, Heidelberg. 

\bibitem[FGR01]{FiGuRai01} Fischer, M. C., Gutierrez-Medina, B. and Raizen, M. G., \emph{Observation of the Quantum Zeno and Anti-Zeno Effects in an Unstable System}, Phys. Rev. Lett., vol. 87, pp. 040-402, 2001.

\bibitem[FQ05]{FQ05} Fagnola, F., and Quezada, R., \emph{Two photon absorption and emision process}, Infinite Dimen. Anal., Quantum Probab. ,Relat. Top. 8, 573, 2005.

\bibitem[FV82]{FV82}A. Frigerio and M. Verri, Long-time asymptotic properties of dynamical semigroups
on $W^*$-algebras, Math. Z. 180, 275-286, 1982.

\bibitem[GT14]{GT14}Gomilko, A.,Tomilov, Y., \emph{On convergence rates in approximation theory for operator semigroups}, Journal of Functional Analysis,
Volume 266, Issue 5, pp. 3040-3082, 2014.

\bibitem[GKS76]{GKS76} Gorini, V., Kossakowski, A., and Sudarshan, E., \emph{Completely positive dynamical semigroups of N-level systems}, Journal of Mathematical Physics. 17, Nr. 5, pp. 821-825, 1976. 

\bibitem[G16]{G} Gl\"uck, J., \emph{A note on approximation of operator semigroups.} Arch. Math. 106, 265–273, 2016.

\bibitem[H06]{H06} Haase, M., \emph{The Functional Calculus for Sectorial Operators}, Operator Theory: Ad-vances and Applications, 169, Birkh\"auser, Basel, 2006.

\bibitem[HRBPK06] {HRBPK06}
Hosten, O., Rakher, M. T., Barreiro, J. T., Peters, N. A., Kwiat,  P. G., \emph{Counterfactual quantum computation through quantum interrogation}, Nature vol. 439, pp. 949–952, 2006.

\bibitem[IHB90]{ItHeBoWi90} Itano, W. M., Heinzen, D. J., Bollinger, J. J. and Wineland, D. J., \emph{Quantum Zeno effect}, Phys. Rev. A, vol. 41, pp. 2295–2300, 1990. 

\bibitem[KT13]{KT13} Kastoryano, M. J. and Temme, K. \emph{Quantum logarithmic Sobolev inequalities and rapid mixing.} Journal of Mathematical Physics, 54(5), 2013.

\bibitem[Ka95]{Kato} Kato, T. \emph{Perturbation Theory for Linear Operators.} vol. 132. Berlin and Heidelberg: Springer, 2 ed. 1995.

\bibitem[KaTz86]{KatzTza86} Katznelson, Y., Tzafiri, L., \emph{On Power Bounded Operators,} Journal of Functional Analysis 68, 3 13-328, 1986.

\bibitem[Ko74]{Kol74} Koliha, J.J., \emph{power-convergence and pseudoinverses of operators in Banach spaces,} Journal of Mathematical Analysis and Applications
Volume 48, Issue 2, pp. 446-469, 1974.

\bibitem[Ko84]{Ko84} Kosaki, H., \emph{Applications of the complex interpolation method to a von Neumann algebra: non-commutative Lp-spaces.} Journal of functional analysis, 56(1), 29–78, 1984.

\bibitem[K85]{Krengel85} Krengel, U., \emph{Ergodic Theorems,} Walter de Gruyter Berlin New York, 1985

\bibitem[L89]{L89} Liggett, T., \emph{ Exponential L2 convergence of attractive reversible nearest particle systems}, Ann. Probab. 17, pp. 403–432, 1989.

\bibitem[Lind76]{Lin76} Lindblad, G., \emph{On the generators of quantum dynamical semigroups}, Communications Mathematical Physics, 48, 119, 1976.

\bibitem[Lin74]{Lin74} Lin, M., \emph{On the uniform ergodic theorem}, Proceedings of the American Mathematical Society, vol. 43, num. 2, 1974.

\bibitem[Ll81]{Llyod81} Llyod, S. P., \emph{On the Uniform Ergodic Theorem Of Lin,} Proceedings of the American Mathematical Society Volume 83, Number 4, 1981.

\bibitem[M04]{Ma04}  Matolcsi, M., \emph{On quasi-contractivity of C0-semigroups on Banach spaces}, Arch. Math., 83, pp. 360–363, 2004.

\bibitem[MS03]{MaShvi03}  Matolcsi, M. and Shvidkoy, R., \emph{Trotter’s product formula for projections}. Arch. Math., 81, pp. 309–317, 2003. 

\bibitem[MS77]{MisSud77} Misra, B., Sudarshan, E. C. G. \emph{The Zenos paradox in quantum theory}, Journal of Mathematical Physics, vol. 18, no. 4, pp. 756–763, 1977.

\bibitem[MW19]{MobWolf19} M\"obus, T. and Wolf, M.M., \emph{Quantum Zeno effect generalized,} Journal of Mathematical Physics, vol. 60, 052201, 2019.

\bibitem[NTY03]{NTY03} Nakazato, H., Takazawa, T. and Yuasa, K.,
\emph{Purification through Zeno-Like Measurements}, Phys. Rev. Lett., vol. 90, 060401, 2003.

\bibitem[NUY04]{NUY04} Nakazato, H., Unoki, M., and Yuasa, K., \emph{Preparation and entanglement purification of qubits through zeno-like measurements.}, Phys. Rev. A, vol .70, 012303, 2004.

\bibitem[OZ99]{OZ99} Olkiewicz, R., and Zegarlinski, B., \emph{ Hypercontractivity in noncommutative Lp spaces.} J. Funct. Anal., 161(1), 246–285, 1999.

\bibitem[PSRDL12]{PSRDL12}
Paz-Silva, G. A., Rezakhani, A. T., Dominy, J. M. and Lidar, D. A.,
\emph{Zeno Effect for Quantum Computation and Control}, Phys. Rev. Lett., vol. 108, 080501, 2012.

\bibitem[RS80]{ReSi80} Reed, M. and Simon, B.. \emph{Methods of Modern Mathematical Physics I: Functional Analysis, Academic Press}, Inc., 1980. 

\bibitem[SHW17]{SieHoWel17} Siemon, I., Holevo, A. S., and Werner, R. F., \emph{Unbounded generators of dynamical semigroups}, \arXiv{1707.02266}

 \bibitem[W12]{Wolf12} Wolf, M., \href{https://www-m5.ma.tum.de/foswiki/pub/M5/Allgemeines/MichaelWolf/QChannelLecture.pdf}{Quantum Channels \& Operations Guided Tour},
 2012.
 
 \bibitem[WYN08]{WYN08}
Wang, X.W., You, J. Q., and Nori, F.,
\emph{Quantum entanglement via two-qubit quantum Zeno dynamics}
Phys. Rev. A, vol. 77, 062339, 2008.
 
\bibitem[Y80]{Yos80} Yosida, K., \emph{Functional Analysis,} 6th edition, Springer Verlag Berlin Heidelberg New York, 1980.

\bibitem[Z17]{Zagr17} Zagrebnov, V. \emph{Comments on the Chernoff$\sqrt{n}$-lemma.} In J. Dittrich, H. Kovarik, and A. Laptev, editors, Functional analysis and operator theory for quantum physics, EMS Series of Congress Reports, pp. 565-573, 2017.
\end{thebibliography}
\end{document}